\documentclass[11pt]{article}
\usepackage[margin=2.5cm]{geometry}

\usepackage{amsthm,amsmath,amssymb}
\usepackage{xcolor}
\usepackage{graphicx}
\usepackage{booktabs}
\usepackage{caption}
\usepackage{hyperref}
\usepackage{algorithm}
\usepackage[noend]{algpseudocode}
\usepackage{enumerate}
\usepackage{enumitem}
\usepackage[normalem]{ulem}
\usepackage{subcaption}
\usepackage{here}
\usepackage{lineno}

\makeatletter
\def\BState{\State\hskip-\ALG@thistlm}
\makeatother

\newcommand{\REMOVE}[1]{}

\renewcommand{\phi}{\varphi}
\newcommand{\leaf}{\text{leaf}}
\newcommand{\down}{\text{down}}

\newcommand{\skel}{\text{skeleton}}
\newcommand{\sw}{\textsf{switch}}

\newtheorem{theorem}{Theorem}
\newtheorem{lemma}[theorem]{Lemma}

\newtheorem{observation}[theorem]{Observation}
\newtheorem{corollary}[theorem]{Corollary}
\newtheorem{claim}[theorem]{Claim}

\newcommand{\ShoLong}[2]{#2}

\newcommand{\FullSODA}[2]{#2}

\renewcommand{\emph}{\textbf}

\title{
Reconfiguration of Connected Graph Partitions\thanks{Research supported in part by NSF CCF-1422311, CCF-1423615, and OIA-1937095. }}
\author{
 Hugo A. Akitaya\thanks{Department of Computer Science, University of Massachusetts Lowell, Lowell, MA, USA.}
\and
Matthew D. Jones\thanks{Khoury College of Computer Sciences, Northeastern University, Boston, MA, USA.}
\and
Matias Korman\thanks{Siemens Electronic Design Automation, Wilsonville, OR, USA.}
\and
Oliver Korten\thanks{Department of Computer Science, Columbia University, New York, NY, USA.}
\and
Christopher Meierfrankenfeld\thanks{Department of Computer Science, Tufts University, Medford, MA, USA.}
\and
Michael J. Munje\thanks{Department of Computer Science, California State University Northridge, Los Angeles, CA, USA.}
\and
Diane L. Souvaine\footnotemark[6]
\and
Michael Thramann\footnotemark[6]
\and
Csaba D. T\'oth\thanks{Department of Mathematics, California State University Northridge, Los Angeles, CA, USA.}~~\footnotemark[6]
}

\begin{document}

\graphicspath{{figs/}}

\newcommand{\contract}{shrink}
\newcommand{\contracted}{shrunk}
\newcommand{\Contract}{Shrink}
\newcommand{\contractible}{shrinkable}
\newcommand{\Contractible}{Shrinkable}
\newcommand{\contractibility}{shrinkability}
\newcommand{\Contractibility}{Shrinkability}
\newcommand{\incontractibility}{unshrinkability}
\newcommand{\Incontractibility}{Unshrinkability}
\newcommand{\incontractible}{unshrinkable}
\newcommand{\Incontractible}{Unshrinkable}

\maketitle
\thispagestyle{empty}

\begin{abstract}
Motivated by recent computational models for redistricting and detection of gerrymandering, we study the following problem on graph partitions. Given a graph $G$ and an integer $k\geq 1$, a \emph{$k$-district map} of $G$ is a partition of $V(G)$ into $k$ nonempty subsets, called \emph{districts}, each of which induces a connected subgraph of $G$. A \emph{switch} is an operation that modifies a $k$-district map by reassigning a subset of vertices from one district to an adjacent district; a \emph{1-switch} is a switch that moves a single vertex.
We study the connectivity of the configuration space of all $k$-district maps of a graph $G$ under 1-switch operations. We give a combinatorial characterization for the connectedness of this space that can be tested efficiently. We prove that it is PSPACE-complete to decide whether there exists a sequence of 1-switches that takes a given $k$-district map into another; and NP-hard to find the shortest such sequence (even if a sequence of polynomial length is known to exist).
We also present efficient algorithms for computing a sequence of 1-switches that takes a given $k$-district map into another when the space is connected, and show that these algorithms perform a worst-case optimal number of switches up to constant factors.
\end{abstract}


\section{Introduction}
\label{sec:intro}
An \emph{electoral district} is a subdivision of territory used in the election of members to a legislative body. \emph{Gerrymandering} is the practice of drawing district boundaries with the intent to give political advantage to a particular group; it 
tends to occur in electoral systems that elect one representative per district. 
Detecting whether gerrymandering has been employed in designing a given district map 
and producing unbiased district maps are important problems to ensure fairness in the outcome of elections. Numerous quality measures have been proposed for the comparison of district maps~\cite{Moon18+,Moon18}, but none of them is known to eliminate bias. Research has focused on exploring the space of all possible district maps that meet certain basic criteria.
Since this space is computationally intractable, even for relatively small instances, randomized algorithms play an important role in finding ``average'' district maps under suitable distributions~\cite{BGH+17}. Being an outlier may indicate that gerrymandering
has been applied in the drawing of a given map~\cite{HRM17}.

Fifield et al.~\cite{MCMC20} model a district map as a vertex partition on an adjacency graph of census tracts or voting precincts. A \emph{census tract} is a small territorial subdivision used as a geographic unit in a census. Each district corresponds to a set of census tracts in the partition and must induce a connected subgraph. The graphs currently used in practice are the dual graphs of a terrain partition, where two districts are adjacent if and only if their boundaries intersect in at least one point. Because of degeneracies, five ``wedge-like'' districts may meet at a single point and induce a $K_5$ in the dual graph.\footnote{Similar phenomenon occurs around a lake, where all districts adjacent to the water are pairwise adjacent.} 
In particular, the district maps are not necessarily planar.

Starting from a given district map, one can obtain another map by switching a subset of census tracts from one district to another. The goal is to apply a sequecne of such operations randomly, and arrive at a uniformly random sample of the space of all possible district maps that meet the desired criteria.
Under some assumptions, Fifield et al.~\cite{MCMC20} proves that the Markov chain produced by their experiments is ergodic\footnote{A Markov chain is {\em ergodic} if it is aperiodic and positive recurrent (that is, each state has a positive probability to be revisited, see~\cite{MCMC20} for more details).}. 
More interestingly, if the assumptions hold, it will have a unique stationary distribution, which is approximately uniform on the space of all $k$-district maps.
One of the assumptions is that the underlying sample space is connected under the switch operation.
However, connectedness is only assumed and remains unproven in~\cite{MCMC20}.

In this paper, we provide a rigorous graph-theoretic background for studying the space of district maps with a given number of districts. We focus on the 1-switch operation that moves precisely one vertex from one district to an adjacent one.
The remainder of the paper will call such an operation simply a ``switch.''
Other than requiring connectedness of districts, we do not impose any other restrictions on the district maps.
In particular, the size of a district can be any integer in the range $[1, n-k+1]$ where $n$ is the number of census tracts. 

The fact that the space of all $k$-district maps is connected in our model implies that any aperiodic Markov chain is also ergodic on the subset of $k$-district maps that meet the desired criteria. 
Thus, our results have implications for models with additional desired criteria (other than connectivity). A natural criterion relevant for the gerrymandering setting is that district maps remain \emph{balanced} (that is, all sets have roughly the same size). 
Besides being an important step in showing theoretical soundness of a Markov-chain-based sampling approach, our results demonstrate how the connectivity of the space relates to how well the adjacency graph is connected. This in turn helps design new operations to traverse the space, and provides a framework for comparing them.

\smallskip\noindent\textbf{Our Results.}
We consider the graph-theoretic model from~\cite{MCMC20}. For an $n$-vertex graph $G$ (the adjacency graph of precincts or census tracts) and an integer $1\leq k\le n$, we consider the \emph{switch graph} $\Gamma_k(G)$ in which each node corresponds to a partition of $V(G)$ into $k$ nonempty subsets (districts), each of which induces a connected subgraph of $G$, and an edge corresponds to switching one vertex from one district to an adjacent district (see Section~\ref{sec:pre} for a definition). 
We do not assume planarity of $G$ unless noted otherwise.

\begin{enumerate}
\item \textbf{Connectedness.} We prove that  $\Gamma_k(G)$ is connected if $G$ is biconnected (Theorem~\ref{thm:2conn-alg}), and give a combinatorial characterization of the connectedness of  $\Gamma_k(G)$ that can be tested in $O(n+m)$ time, where $n=|V(G)|$ and $m=|E(G)|$ (Theorem~\ref{thm:conn-test}).
    In general, however, it is PSPACE-complete to decide whether two given nodes of $\Gamma_k(G)$ are in the same connected component even when $G$ is planar (Theorem~\ref{thm:planar pspace}), or $G$ is nonplanar and $k=2$ (Theorem~\ref{thm:two district pspace}).
\item \textbf{\Contractible\ Districts.} One of our key methods to modify a district map is to \contract\ a district into a single vertex by a sequence of switch operations. If this is feasible, we call the district \emph{\contractible}; if all districts are \contractible, we call the district map \emph{\contractible}.
We prove that the subgraph $\Gamma_k'(G)$ of $\Gamma_k(G)$  induced by \contractible\ district maps is connected if $G$ is connected (Theorem~\ref{thm:general-graphs}).
\item \textbf{Diameter.} When $G$ is biconnected, the diameter of $\Gamma_k(G)$ is in $O(kn)$, 
where $n=|V(G)|$ (Theorem~\ref{thm:general-graphs}), and this bound is the best possible  (Theorem~\ref{thm:contractible-LB+}).
When $\Gamma_k(G)$ is disconnected, the diameter of a component may be as large as $2^{\Omega(n)}$ even for planar graphs (Corollary~\ref{corollary:exp switch diameter}).
\item \textbf{Shortest Path.} Finding the distance between two nodes of $\Gamma_k(G)$ is NP-hard, even if $\Gamma_k(G)$ is connected  (Theorem~\ref{thm:biconnected-hardness}).
\end{enumerate}

\ShoLong{Due to lack of space, portions of the document have been moved to the Appendix.}{}
\smallskip

\noindent\textbf{Related Previous Work.}
Graph partitions and graph clustering algorithms~\cite{GraphPartition} are widely used in divide-and-conquer strategies. These algorithms, however, do not explore the space of all partitions into $k$ connected subgraphs. Evolutionary algorithms~\cite{Sanders2012}, in this context, modify a partition by random ``mutations,'' which are successive coarsening and uncoarsening operations, rather than moving vertices from one subgraph to another.

While the adjacency graph model for district maps has been used for decades in combinatorial optimization and operations research~\cite{RiccaSS13}, the objective was finding optimal district maps under one or more criteria. Since exhaustive search is infeasible and most variants of the optimization problem are intractable~\cite{PuppeT09}, local search heuristics were suggested~\cite{RiccaS08}.
Several combinatorial results restrict $G$ to be a square grid~\cite{ApollonioBLRS09,PuppeT08}.
Heuristic and intractability results are also available for geometric variants of the optimization problem,
where districts are polygons in the plane~\cite{FleinerNT17,KMV-hardness-19,RiccaSS08}.

Elementary graph operations similar to our 1-switch operation have also been studied.
Motivated by the classical ``Fifteen'' puzzle, Wilson~\cite{WILSON1974} studied the configuration space of $t$, $t< n$, indistinguishable pebbles (a.k.a.~tokens~\cite{MonroyFHHUW12}) on the vertices of a graph $G$ with $n$ vertices, where each pebble occupies a unique vertex of $G$, and can move to any adjacent unoccupied vertex.
The occupied and unoccupied vertices partition $V(G)$ into two subsets. Crucially, the number of pebbles is fixed, and the occupied vertices need not induce a connected subgraph. Results include a combinatorial characterization of the configuration space (a.k.a.~token graph)~\cite{WILSON1974}, NP-hardness for deciding connectedness~\cite{KornhauserMS84}, finding the shortest path between two configurations~\cite{Goldreich11,RatnerW90}, and bounds on the diameter and the connectivity of the configuration space~\cite{LeanosT18,MonroyFHHUW12}. Demaine et al.~\cite{DemaineDFHIOOUY15}
considered a subgraph of the token graph, where the tokens are located at an independent set.
The diameter and shortest path in the configuration space can often be computed efficiently when
the underlying graph $G$ is a tree~\cite{AulettaMPP99,DemaineDFHIOOUY15}, or chordal~\cite{BonamyB17}.
There are a few results that require the occupied vertices to induce a connected subgraph, but they are limited to the case where $G$ is a grid~\cite{DumitrescuP06,KomuravelliSB09}, and the number of pebbles is still fixed.

Goraly et al.~\cite{GoralyH10} later considered colored pebbles (tokens). Each color class consists of indistinguishable pebbles, unoccupied vertices are considered as one of the color classes ~\cite{FujitaNS12,YamanakaDIKKOSS15,YamanakaHKKOSUU18}: Hence all vertices in $V(G)$ are occupied and a move can swap the pebbles on two adjacent vertices. The color classes (including the ``unoccupied'' color) partition $V(G)$ into subsets. However, the cardinality of each color class remains fixed and the color classes need not induce  connected subgraphs. Results, again, include combinatorial characterizations to connected configuration space~\cite{FoersterGHKSW17}, NP-completeness for the connectedness of the configuration space for $k\geq 3$ colors, and a polynomial-time algorithm for finding the shortest path for $k=2$ colors. See~\cite{BonnetMR18,MiltzowNORTU16} for recent results on the parametric complexity of these problems.

The problem of partitioning a graph $G$ into $k$ connected subgraphs with equal (or almost equal) number of vertices is known as the \emph{Balanced Connected $k$-Partition Problem} (BCP$_k$), which is NP-hard already for $k=2$~\cite{DyerF85}, for grids in general~\cite{BerengerNP18}, and also hard to approximate within an absolute error of $n^{1-\delta}$~\cite{Chlebikova96}. 

\smallskip\noindent\textbf{Organization.}
Section~\ref{sec:pre} defines the reconfiguration problem formally, and describes some important properties of \contractible\ districts. Section~\ref{sec:alg} shows that $\Gamma_k(G)$ is connected if $G$ is biconnected, and $\Gamma_k'(G)$ is connected if $G$ is connected. Section~\ref{ssec:hardness-connect} presents our PSPACE-completeness proof and lower bounds for the diameter of $\Gamma_k(G)$, and Section~\ref{ssec:hardnes-shortest} continues with our NP-hardness results for the shortest path problem.
\ShoLong{}{We conclude in Section~\ref{sec:con} with open problems.}

\section{Preliminaries}
\label{sec:pre}


Let $G=(V,E)$ be a connected graph. A \emph{$k$-district map} $\Pi$ of $G$ is a partition of $V(G)$ into disjoint nonempty subsets $\{V_1,\ldots,V_k\}$ such that the subgraph induced by $V_i$ is connected for all $i\in\{1,\ldots,k\}$.
Each subgraph induced by $V_i$ is called a \emph{district}.
We abuse the notation by writing $\Pi(v)$ for the subset in $\Pi$ that contains vertex $v$. We now formally define the switch operation. Our definition matches the previous informal description. Given a $k$-district map $\Pi=\{V_1,\ldots,V_k\}$, and a path $(u,v,w)$ in $G$ such that $\Pi(u)=\Pi(v)\neq \Pi(w)$, a \emph{switch} (denoted \sw$_\Pi(u,v,w)$) is an operation that returns a $k$-district map obtained from $\Pi$ by removing $v$ from the subset $\Pi(u)$ and adding it to $\Pi(w)$.
More formally, 
\[\sw_\Pi(u,v,w)=\Pi'=(\Pi\setminus\{\Pi(u),\Pi(w)\})\cup\{\Pi(u)\setminus\{v\},\Pi(w)\cup\{v\}\}\] if $ \Pi'$ is a $k$-district map.
Note that \sw$_\Pi(u,v,w)$ is not defined if $\Pi(v)\setminus\{v\}$ induces a disconnected subgraph.
A switch is always reversible since if \sw$_\Pi(u,v,w)=\Pi'$, then \sw$_{\Pi'}(w,v,u)=\Pi$.
We may omit the subscript when the map in which the switch is applied is clear from context.
For every graph $G$ and integer $k$, the \emph{switch graph} $\Gamma_{k}(G)$
is the graph whose vertex set is the  set of all $k$-district maps of $G$, and $\Pi_1, \Pi_2\in V(\Gamma_{k}(G))$ are connected by an edge if there exist $u,v,w\in V(G)$ such that \sw$_{\Pi_1}(u,v,w)=\Pi_2$.

\global\def \preliminarysec {
\label{ssec:2con}
Biconnectivity plays an important role in our proofs. In particular, we rely on the concept of a \emph{block tree}, which represents the containment relation between the blocks (maximal biconnected components) and the cut vertices of a connected graph, and a \emph{SPQR tree}, which is a recursive decomposition of a biconnected graph. We review both concepts here.

\begin{figure}[htpb]
\begin{center}
\includegraphics[width=\textwidth]{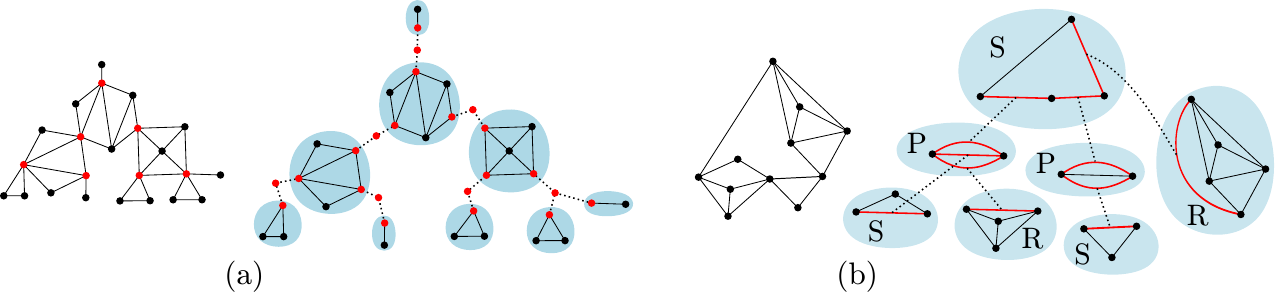}
\caption{(a) A connected graph and its block tree. Cut vertices are shown in red, and dotted lines connect two occurrences of the same vertex in adjacent nodes of the block tree. (b) A 2-connected graph and its SPQR tree. Virtual edges are shown in red. Corresponding pairs of virtual edges are connected with dotted lines.}
\label{fig:block-spqr}
\end{center}
\end{figure}

\smallskip\noindent\textbf{Block Trees.} Let $G$ be a connected graph. Let $B(G)$ be the set of \emph{blocks} of $G$. (Two adjacent vertices induce a 2-connected subgraph, so a block may be a subgraph with a single edge.) Let $C(G)$ be the set of cut vertices in $G$. Then the \textbf{block tree} $T=T(G)$ is a bipartite graph, whose vertex set is $V(T)=B(G)\cup C(G)$, and $T$ contains an edge $(W,c)\in B(G)\times C(G)$ if and only if $c\in W$ (i.e., block $W$ contains vertex $c$). The definition immediately implies that a leaf and its unique neighbor induce a block $W\in B(G)$ (and never a cut vertex in $C(G)$). The block tree can be computed in $O(|E(G)|)$ time and space~\cite{HopcroftT73}.
For convenience, we label every biconnected component by its vertex set (i.e., for a block $W\in B(G)$, we denote by $W$ the set of vertices in the block).

\smallskip\noindent\textbf{SPQR Trees.}  Let $G$ be a biconnected graph. A deletion of a (vertex) 2-cut $\{u,v\}$ disconnects $G$ into two or more components $C_1,\ldots ,C_i$, $i\geq 2$. A \textbf{split component} of $\{u,v\}$ is the subgraphs of $G$ induced by $V(C_j)\cup\{u,v\}$ for $j=1,\ldots, i$, or the graph induced by $\{u,v\}$ if $uv\in E(G)$.
The SPQR-tree $T_G$ of $G$ represents a recursive decomposition of $G$ defined by its 2-cuts.
A node $\mu$ of $T_G$ is associated with a multigraph called \emph{$\skel(\mu)$} on a subset of $V(G)$ obtained by adding a \textbf{virtual} edge $uv$ to a split component of the 2-cut $\{u,v\}$, or by creating a virtual (parallel) edge $uv$ for each split component of $\{u,v\}$. 
Hence, an edge in $\skel(\mu)$ is \textbf{real} if it is an edge in $G$, or virtual otherwise.
A node $\mu$ has a \textbf{type} in $\{$S,P,R$\}$.
If the type of $\mu$ is S, then $\skel(\mu)$ is a cycle of 3 or more vertices.
If the type of $\mu$ is P, then $\skel(\mu)$ consists of 3 or more parallel edges between a pair of vertices.
If the type of $\mu$ is R, then $\skel(\mu)$ is a 3-connected graph on 4 or more vertices.
Two nodes $\mu_1$ and $\mu_2$ of $T_G$ are adjacent if $\skel(\mu_1)$ and $\skel(\mu_2)$ share exactly two vertices, $u$ and $v$, that form a 2-cut in $G$.
Each virtual edge in $\skel(\mu)$
has a corresponding pair in $\skel(\mu')$ for some adjacent node $\mu'$; see Figure~\ref{fig:block-spqr}(b).
The graph $G$ can be reconstructed from the skeletons of the nodes in $T_G$ by identifying every pair of corresponding virtual edges and then deleting all virtual edges. No two S nodes (resp., no two P nodes) are adjacent. Therefore, $T_G$ is uniquely defined by $G$.
If $\mu$ is a leaf in $T_G$, then $\skel(\mu)$ has a unique virtual edge; in particular the type of every leaf is S or R.
The SPQR tree $T_G$ has $O(|E(G)|)$ nodes and can be computed in $O(|E(G)|)$ time~\cite{BattistaT96}. 

}

\subsection{Block Trees and SPQR Trees}
\preliminarysec

\subsection{\Contractibility}

Consider a graph $G$ and a $k$-district map $\Pi$. We say that the operation \sw$_\Pi(u,v,w)$
\emph{\contract s} $\Pi(u)$ to $\Pi(u)\setminus\{v\}$, and \emph{expands} $\Pi(w)$ to $\Pi(w)\cup\{v\}$.
%
%
A sequence of switches \emph{\contract s} (resp., \emph{expands}) $V_i$ to $V_i'$ if there exists a sequence of consecutive switches that jointly \contract\ (resp., expand) $V_i$ to $V_i'$. A subset $V_i\in\Pi$ (and its induced district) is \emph{\contractible} if it can be \contracted\ to a singleton (district of size one)
by a sequence of $|V_i|-1$ switches; otherwise it is \emph{\incontractible}.
A $k$-district map is \emph{\contractible} if each of its districts is \contractible.
A district $V_i$ is said to \emph{contain} a block $W\in B(G)$ if it contains all vertices in $W$.

In the remainder of this section we state some simple properties that will be used later. 
\begin{lemma}\label{lem:leaves}
A switch operation cannot move a leaf of $G$ from one district to another.
\end{lemma}
\begin{proof}
Let $v\in V(G)$ be a leaf in $G$, and let $u\in V(G)$ be its unique neighbor.
Since $v$ is a leaf there is no path $(u,v,w)$ and hence there is no valid \sw$_\Pi(u,v,w)$ moving $v$ to another district.
\end{proof}

\begin{lemma}\label{lem:con1}
Let $T$ be the block tree of a graph $G$, and let $\Pi$ be a $k$-district map on $G$. If a district $V_\ell$ contains two leaves of $T$, say $W_i, W_j \in B(G)$, then a switch operation cannot move any vertex from $W_i\cup W_j$ to another district. Consequently, $V_\ell$ is \incontractible.
\end{lemma}
\begin{proof}
Suppose, for the sake of contradiction, that $W_i\cup W_j\subseteq V_\ell$ and a switch moves some vertex $v\in W_i\cup W_j$ to another district. 
Since $W_i$ and $W_j$ are leaves in $T$, only their cut vertices can be adjacent to vertices outside of $W_i\cup W_j$. 
Then, $v$ is a cut vertex in $\Pi(v)$ and $\Pi(v)\setminus \{v\}$ does not induce a connected subgraph in $G$, a contradiction.

Since $W_i \ne W_j$, there are at least two vertices, one from each block, that remain in $V_\ell$ after any sequence of switch operation. Consequently, $V_\ell$ cannot become a singleton.
\end{proof}

\begin{lemma}\label{lem:con2}
Let $\Pi$ be a $k$-district map on $G$ for some $k\geq 2$, and let $V_i\in \Pi$ such that $V_i$ contains at most one leaf of the block tree $T$ of $G$. Then $V_i$ is \contractible. Furthermore,
\begin{itemize}
	\item if $V_i$ does not contain any leaf of the block tree, then $V_i$ can be \contracted\ to any of its vertices;
	\item if $V_i$ contains a leaf $W_j\in B(G)$ of the block tree, then $V_i$ can be \contracted\ to a vertex $v$ if and only if $v\in W_j$ and $v$ is not the parent cut vertex of $W_j$.
\end{itemize}
In both cases, a sequence of $|V_i|-1$ switches that \contract\ $V_i$ can be computed in $O(|E(G)|)$ time.
\end{lemma}

\global\def \proofLemConTwo {
\begin{proof}
We first prove a necessary condition for \contract ing a district to a target vertex. Assume that $V_i$ can be \contracted\ to a vertex $t\in V_i$, and $V_i$ contains exactly one leaf $W_j\in B(G)$ of the block tree. Let $c_j$ be the parent cut vertex of $W_j$. Since every path between $W_j\setminus \{c_j\}$ and $V_i\setminus W_j$ contains $c_j$, no vertex in $W_j\setminus \{c_j\}$ can change districts until $c_j$ and all vertices of $V_i$ outside of $W_j$ have switched to some other districts. At this point, we have $V_i=W_j\setminus \{c_j\}$, consequently $t\in W_j\setminus \{c_j\}$, as required.

We next show that the above conditions are sufficient. Assume that $V_i$ and a target vertex $t\in V_i$ satisfy the above restrictions. It is enough to show that if $V_i\neq \{t\}$, there exists a vertex $v\in V_i\setminus \{t\}$, such that $v$ can be switched to another district; and $t$ and $V_i\setminus \{v\}$ satisfy the conditions above. Then we can successively switch all vertices in $V_i\setminus \{t\}$ to other districts until $V_i=\{t\}$.
 
Let $G'$ be the subgraph induced by $V_i$. Compute the block tree of $G'$, and denote it by $T'$. Root $T'$ at the block vertex in the tree that contains $t$. We distinguish between cases.

\begin{itemize}
\item If $G'$ is not biconnected, then $G'$ contains two or more leaf blocks. Let $W'\in B(G')$ be a leaf block in $T'$ other than the root, and let $c'\in C(G')$ be its parent cut vertex. Note that $W'$ is not a leaf block in $T$, otherwise $V_i$ would contain this leaf block, 
contradicting our assumptions.
Then, it is either a subset of a nonleaf block of $T$ or a proper subset of a leaf block of $T$. In either case, there exists a vertex $v\in W'\setminus \{c'\}$ adjacent to some vertex $u\not\in V_i$. Since $W'$ is biconnected, $W'\setminus \{v\}$ induces a connected subgraph in $G$; consequently $V_i\setminus \{v\}$ induces a connected subgraph, as well. Therefore, $v$ can be switched to the district of $u$.
\item If $G'$ is biconnected, then $G'$ is a subgraph of some block $W\in B(G)$.
    We claim that there exists a vertex $v\in V_i\setminus \{t\}$ adjacent to some vertex $u\not\in V_i$. To prove the claim, suppose the contrary. Then every path between $V_i\setminus \{t\}$ and $V(G)\setminus V_i$ goes through $t$. This implies that $t$ is a cut vertex, and $V_i$ is a leaf block in $T$, which contradicts our assumption. Now again, $v$ can be switched to the district of $u$.
\end{itemize}
First, note that the switch operation maintain the property that $V_i$ contains at most one leaf block of $T$. Indeed, since we  \contract\ $V_i$, the number of leaf blocks contained in $V_i$ monotonically decreases. 
Second, we show that $t$ remains a valid choice for the target vertex. If $V_i\setminus \{v\}$ contains the same leaf blocks as $V_i$, then $t$ remains a valid target. Otherwise $V_i$ contains a leaf block, say $W_j$, and $V_i\setminus \{v\}$ does not, then $v$ is the parent cut vertex of $W_j$. In this case, $t\in V_i\setminus \{v\}$, and any vertex in $V_i\setminus \{v\}$ is a valid choice for $t$. 
This proves that $V_i$ is \contractible, as required.

Our proof is constructive and leads to an efficient algorithm that successively switches every vertex in $V_i\setminus \{t\}$ to some other districts until $V_i=\{t\}$. The block trees $T$ and $T'$ can be computed in $O(|E(G)|)$ time~\cite{HopcroftT73}. While $V_i$ is \contracted, we maintain the induced subgraph $G'$, and the set of edges between $V_i$ and $V(G)\setminus V_i$ in $O(|E(G)|)$ total time. While $T'$ contains two or more blocks, we can successively switch all vertices of a leaf block $W'$ that does not contain $t$ to other districts; eliminating the need for recomputing $T'$. Then, the total running time is $O(|E(G)|)$.
\end{proof}
}
\ShoLong{
\begin{proof}[Proof sketch]
The full proof can be found in Section~\ref{sec:proofLemConTwo}.
Assume there exists a sequence of switches that \contract s $V_i$ into $\{t\}$ where $t$ is a vertex in $V_i$.
If $V_i$ contains a leaf block $W_j$ with cut vertex $c_j$, then the vertices in  
$W_j\setminus\{c_j\}$ are not adjacent to any other district. Therefore, district $V_i$ cannot switch out of a vertex $v\in W_j\setminus\{c_j\}$ until $V_i$ is contained in $W_j\setminus\{c_j\}$, otherwise $c_j$ would be a cut vertex in the subgraph induced by $V_i$. Hence, $t\in W_j\setminus\{c_j\}$.

If $t$ is in $W_j\setminus\{c_j\}$ or $V_i$ does not contain any leaf block, then there always exists some $v\in V_i\setminus\{t\}$ such that $v$ can be switched out of $V_i$  until $V_i=\{t\}$. We can find $v$ by computing the block tree of the subgraph $G'$ induced by $V_i$. Assume such a tree has a leaf block that does not contain $t$. Then, any vertex in this leaf block, other than the cut vertex, that is adjacent to some other district can be switched out of $V_i$. If $G'$ is biconnected (i.e., its block tree is a singleton), then there exists a vertex $v\in V_i\setminus\{t\}$ adjacent to another district since $G'$ is a subgraph of a block of $G$ and $t$ is not a cut vertex. Then, we can switch $v$ out of $V_i$.

\end{proof}
}
{\proofLemConTwo}

\begin{lemma}\label{lem:invariant}
The \contractibility\ (resp., \incontractibility) of a $k$-district map on a graph $G$ is invariant under switch operations.
\end{lemma}
\begin{proof}
Every \incontractible\ $k$-district map contains some \incontractible\ district $V_\ell$. Lemmas~\ref{lem:con1}--\ref{lem:con2} show that a subset $V_\ell \in \Pi$ is \incontractible\ if and only if $V_\ell$ contains at least two leaves of the block tree, say $W_i,W_j\subset V_\ell$.
By Lemma~\ref{lem:con2}, $W_i\cup W_j\subseteq V_\ell$ after any sequence of switches, so $V_\ell$ remains \incontractible.
The rest of the proof is implied by the reversibility of switches.
\end{proof}

\section{Connectedness}
\label{sec:alg}

\ShoLong{}{In this section we characterize graphs $G$ for which the switch graph $\Gamma_k(G)$ is connected. We give two results depending on the connectivity of $G$. 

}

\subsection{Biconnected Graphs}
\label{ssec:2conn-alg}


\begin{theorem}\label{thm:2conn-alg}
For every biconnected graph $G$ with $n$ vertices, and for every integer $1\leq k\le n$,
the switch graph $\Gamma_k(G)$ is connected and its diameter is bounded by $O(kn)$.
\end{theorem}

\begin{proof}
We may assume that $1<k<n$, otherwise $\Gamma_k(G)$ is trivially connected. We present an algorithm (Algorithm~\ref{algo:2conn}) that performs a sequence of switches that transform $\Pi$ into a canonical $k$-district map of $G$, that we denote by $\Pi_0$. 
We show that $\Pi_0$ depends only on $G$ and $k$ (but not on $\Pi$). 
Consequently, any two $k$-district maps can be transformed to $\Pi_0$, and $\Gamma_k(G)$ is connected.

\begin{algorithm}[htbp]
\caption{Canonical Algorithm for Biconnected Graphs}\label{algo:2conn}
\begin{algorithmic}[1]
\Procedure{Canonical}{$G,k,\Pi$}
\While {$k>1$}
\State Compute the SPQR tree $T_G$ of $G$; order the leaves by DFS; let $\mu$ be the first leaf.
\If{$\mu$ is an S node (and $\skel(\mu)$ is a cycle with one virtual edge)}
\State Let $\skel(\mu)=(v_1,\ldots, v_t)$, where $v_1v_t$ is the virtual edge; set $i=2$.
    \While {$i<t$ and $k>1$}
    \State \Contract\ $\Pi(v_i)$ to $\{v_i\}$.
    \State Delete vertex $v_i$ from $G$, and put $i:=i+1$ and $k:=k-1$.
    \EndWhile
\Else{ $\mu$ is an R node (and $\skel(\mu)$ is triconnected)}
\State Let $v$ be an arbitrary vertex that is not incident to the (unique) virtual edge.
\State \Contract\ $\Pi(v)$ to $\{v\}$.
\State Delete vertex $v$ from $G$, and put $k:=k-1$.
\EndIf
\EndWhile
\EndProcedure
\end{algorithmic}
\end{algorithm}

\smallskip\noindent\textbf{Proof of Correctness.}
Algorithm~\ref{algo:2conn} successively \contract s a district into a single vertex,
and then deletes this vertex from the graph, and the corresponding district from $\Pi$,
until the number of districts drops to 1. We need to show that each district that the algorithm
\contract s into a singleton is \contractible. We prove an invariant that imply this property:

\global\def \pfclaimAAA {
\begin{proof}
Let $\mu$ be the leaf node in line~3 of the algorithm.
If $\mu$ is an R node, then the $G$ remains biconnected after the deletion of a vertex $v$.
Assume that $\mu$ is an S node, corresponding to a cycle $(v_1,\ldots , v_t)$, $t\geq 3$, where $v_1v_t$ is the only edge that corresponds to a virtual edge. Then the deletion of all vertices in $\{v_1,\ldots, v_{t-1}\}$ produces a biconnected graph; and the deletion of a $\{v_2,\ldots ,v_i\}$, $2\leq i<t-1$, produces a biconnected graph with a ``dangling'' path $(v_{i+1},\ldots, v_t)$.
\end{proof}
}
\global\def \claimAAA {
The graph $G$ remains connected during Algorithm~\ref{algo:2conn}.
}

\begin{claim}\label{cl:BBB}
$G$ remains connected and the district map $\Pi$ remains \contractible\ during Algorithm~\ref{algo:2conn}.
\end{claim}

\global\def \pfclaimBBB {
\begin{proof}
In a biconnected graph, every district is \contractible\ by Lemma~\ref{lem:con2}.
Let $\mu$ be the leaf node in line~3 of the algorithm.
If $\mu$ is a R node, then the graph $G$ remains biconnected after the deletion of a vertex, and so the $(k-1)$-district map of the remaining graph is \contractible.
If $\mu$ is an S node, then $G$ obtained by deleting vertex $v_i$ is either biconnected or a biconnected graph with a ``dangling'' path $(v_{i+1},\ldots, v_t)$. In both cases, $G$ has at most one leaf block (namely, a 1-edge block at the end of the dangling path). By Lemma~\ref{lem:con2}, every district that contains at most one leaf block is \contractible, and so the district map remains \contractible.
\end{proof}
}
\ShoLong{}{\pfclaimBBB}


The following claim establishes that the switch graph $\Gamma_k(G)$ is connected since it contains a path from any district map to the district map produced by Algorithm~\ref{algo:2conn}.

\begin{claim} \label{cl:CCC}
The district map $\Pi_0$ depends only on $G$ and $k$.
\end{claim}

\global\def \pfclaimCCC {
\begin{proof}
The map $\Pi_0$ contains the deleted singleton districts and one larger district.
Since each vertex deleted from the graph $G$ was selected based on the current graph $G$, its SPQR tree, and the DFS order of its leaves, the sequence of deleted vertices depends only on $G$ and $k$. 
\end{proof}
}
\ShoLong{}{\pfclaimCCC}


\smallskip\noindent\textbf{Analysis.}
Algorithm~\ref{algo:2conn} successively \contract s $k-1$ districts into singletons. By Lemma~\ref{lem:con2}, for each district this is done by a sequence of $O(n)$ switches that can be computed in $O(|E(G)|)$ time.
Overall Algorithm~\ref{algo:2conn} runs in $O(k|E(G)|)$ time and performs $O(kn)$ switch operations.
For any two $k$-district maps, $\Pi_1$ and $\Pi_2$, there exists a sequence of $O(kn)$ switches that takes $\Pi_1$ to $\Pi_0$ and then to $\Pi_2$. Therefore, the diameter of $\Gamma_k(G)$ is $O(kn)$.
\end{proof}

The following theorem shows that the upper bound in Theorem~\ref{thm:2conn-alg} is asymptotically tight.

\begin{theorem}\label{thm:contractible-LB+}
For all integers $1\leq k\leq n$, there exists a biconnected graph $G$ with $n$ vertices
such that the diameter of  $\Gamma_k(G)$ is $\Omega(k(n-k))$.
\end{theorem}

\begin{proof}
Let $G=C_n$ be the cycle with $n$ vertices $(v_1,\ldots, v_n)$.
We construct two $k$-district maps, $\Pi_1$ and $\Pi_2$. Let $\Pi_1$ consist of $V_i=\{v_i\}$ for $i=1,\ldots , k-1$, and $V_k=\{v_k,\ldots, v_n\}$. The partition $\Pi_2$ is the copy of $\Pi_1$ rotated by $\lfloor n/2\rfloor$, that is, $V_i'=\{v_{i+\lfloor n/2\rfloor}\}$ for $i=1,\ldots , k-1$, and $V_k'=\{v_{k+\lfloor n/2\rfloor},\ldots, v_{n+\lfloor n/2\rfloor}\}$, where we use arithmetic modulo $n$ on the indices.

Assume that a sequence of switch operations takes $\Pi_1$ to $\Pi_2$. Note that the cyclic order of the district cannot change,
and so there is an integer $r\in \{0,\ldots , k-1\}$ such that $V_i$ is transformed to $V_{i+r\mod k}'$ for all $i\in \{1,\ldots ,k\}$.
For any $r$, at least $k-2$ districts are singletons in both $\Pi_1$ and $\Pi_2$. The sum of the shortest distances
along $C_n$ between the initial and target positions is a lower bound for the number of switches.

If $r\leq \lfloor k/2\rfloor$, then the shortest distance between the initial and target positions is at least $\lfloor n/2\rfloor -r\in \Omega(n-k)$ for the districts $V_i$, $i=1\ldots, k-1-r$; which sums to $\Omega(k(n-k))$.
If $\lfloor k/2\rfloor <r<k$, then shortest distance is at least $\lfloor n/2\rfloor -(k-r)\in \Omega(n-k)$ for $V_i$, $i=r,\ldots, k-1$; which also sums to $\Omega(k(n-k))$.
\end{proof}

\subsection{Algorithm for General Connected Graphs}
\label{ssec:general-graphs}

Recall that $\Gamma_k'(G)$ is the subgraph of $\Gamma_k(G)$ induced by \contractible\ district maps.
If $G$ is a biconnected graph, then every district map is \contractible\ by Lemma~\ref{lem:con2}, and so $\Gamma_k(G)=\Gamma_k'(G)$.
In this section, we extend this result to a larger family of graphs, showing that if $G$ is connected, then $\Gamma_k'(G)$ is connected. That is, any \contractible\ $k$-district map can be carried to any other shrinkable $k$-district by a sequence of switch operations.

\begin{theorem}\label{thm:general-graphs}
For every connected graph $G$ with $n$ vertices, and for every integer $1\leq k\le n$,
the switch graph $\Gamma_k'(G)$ over \contractible\ $k$-district maps
is connected and its diameter is $O(kn)$.
\end{theorem}

A crucial technical step is to move a district from one block to another, through a cut vertex. This is accomplished in the following technical lemma.

\begin{lemma}\label{lem:pushing}
Let $G$ be a connected graph 
whose block tree contains at least
two 
blocks, $W_1,W_2\in B(G)$, and let $P$ be a shortest path from a vertex in $W_1$ to a vertex in $W_2$ (possibly, $P$ has a single vertex). Let $\Pi$ be a district map of $G$ in which each vertex of $P$ is a singleton district, but $W_1$ contains a district of size more than one. Then there is a sequence of 
$O(|W_1|+|P|)$ switches that increases the number of districts in $W_1$ by one, and decreases the number districts in $W_2$ by one.
\end{lemma}
\begin{proof}
Let $c_1\in W_1$ and $c_2\in W_2$ be the two endpoints of $P$; possibly $c_1=c_2$. 
Note that $c_1, c_2\in C(G)$ since $P$ is a shortest path between $W_1$ and $W_2$.
We claim that after $O(|W_1|)$ switch operations in $W_1$, we can find a path $P^*=(p_0,p_1,\ldots ,p_m)$ such that $\{p_0,p_1\}$ is a 2-vertex district in $W_1$, all other vertices in $P^*$ are singleton districts, and $P=(p_\ell,\ldots, p_m)$ for $1\leq \ell\leq m$ (with $p_\ell=c_1$ and $p_m=c_2$). Assuming that this is possible, we can then successively perform $\sw(p_{i-1},p_i,p_{i+1})$ for $i=1,\ldots, m-1$, which replaces $\{p_0,p_1\}$ by two singleton districts, and produces a 2-vertex district $\{p_{m-1},p_m\}$. Finally, we \contract\ this district to $\{p_{m-1}\}$ by Lemma~\ref{lem:con2}, thereby decreasing the number of districts in $W_2$ by one. Overall, we have used $O(|W_1|+|P^*|)=O(|W_1|+|P|)$ switches.

To prove the claim, let $G_1$ be the biconnected subgraph of $G$ induced by $W_1$. Let $Q=(q_1,\ldots , q_s)$ be a shortest path between $q_s=c_1$ and a vertex in a district $V_0\subseteq W_1$ of size $|V_0|>1$. Since $Q$ is a shortest path, the vertices $q_2,\ldots, q_s$ are singleton districts. If $|V_0|=2$, say $V_0=\{q_0,q_1\}$, then we can take $P^*=(q_0,q_1,\ldots ,q_s)\oplus P$, where $\oplus$ is the concatenation operation.  

Assume that $|V_0|>2$. Since $G_1$ is biconnected, $V_0$ can be \contracted\ to $\{q_1\}$ by a sequence of $|V_0|-1=O(|W_1|)$ switches by Lemma~\ref{lem:con2}. Each switch in the sequence \contract s $V_0$ and expands some adjacent district. Perform the switches in this sequence until either (a) $|V_0|=2$, or (b) some singleton district $\{q_i\}$, $i=2,\ldots, s$, expands.
In both cases, we find a path $Q'=(q_i,\ldots , q_s)$, $i\in \{1,\ldots, s\}$, such that $q_i$ is in some 2-vertex district $\{q_0,q_i\}$, all other vertices in $Q'$ are singletons, and $q_s=c_1$. Consequently, we can take $P^*=(q_0,q_i,\ldots ,q_s)\oplus P$, as claimed.
\end{proof}

We can now consider the general case.
Let $G$ be a connected graph with $n$ vertices and let $1\leq k\leq n$. We present an algorithm (Algorithm~\ref{algo:1conn}) that transforms a given \contractible\ $k$-district map $\Pi$ into one in pseudo-canonical form (defined below), and then show that any two $k$-district maps in pseudo-canonical form can be transformed to each other. Consequently, any two \contractible\ $k$-district maps can be transformed into each other, and $\Gamma_k'(G)$ is connected.

\begin{figure}[htpb]
\begin{center}
\includegraphics[width=0.99\textwidth]{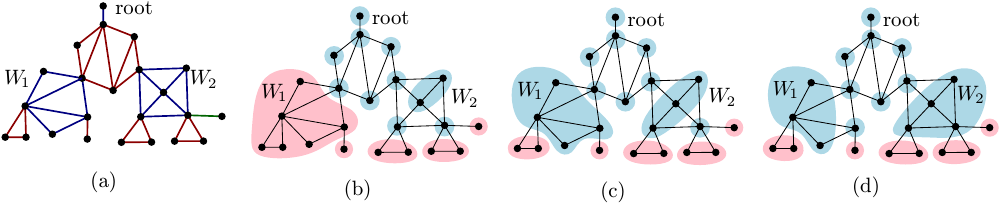}
\caption{(a) A connected graph with nine blocks. (b) A pseudo-canonical 15-district map. Five leaf districts are red and ten nonleaf districts are blue. (c) and (d): Pseudo-canonical district maps obtained from (b) by moving districts from $W_2$ to $W_1$ by successive applications of Lemma~\ref{lem:pushing}.}
\label{fig:pseudo-canonical}
\end{center}
\end{figure}

We introduce some additional terminology; see Figure~\ref{fig:pseudo-canonical}~(a). Let $T$ be a block tree of $G$. Fix an arbitrary leaf block $R\in B(G)$. We consider $T$ as an ordered tree, rooted at $R$, where the children of each node are ordered arbitrarily. For a district map $\Pi$, we define a \textbf{leaf district} to be a district that contains every vertex of some nonroot leaf block $W\in B(G)$, with the possible exception of its parent cut vertex $c\in C(G)$. 
Note that a leaf district could have vertices outside the leaf block. Moreover, every leaf district $V_i$ corresponds to a unique leaf block (otherwise $\Pi$ would be \incontractible\ by Lemma~\ref{lem:con1}), and we denote this block by $\leaf(V_i)$. A leaf district is \contractible\ into any vertex in $\leaf(V_i)$, except for $c$ (cf.~Lemma~\ref{lem:con2}). 
Further note that a district may become a leaf district over the course of the algorithm, while leaf districts remain leaf districts.

For every block $W\in B(G)$, except for the root, we define a set $\down(W)$ as follows. Let $c\in C(G)$ be the parent of $W$ in $T$, let $V_i$ be the district that contains $c$, and let $\down(W)$ be the set of vertices in $V_i$ that lie in $W$ or its descendants.
The set $\down(W)$ is an \textbf{elbow} if $\down(W)\neq \{c\}$, $V_i$ is a leaf district, and $\down(W)$ does not contain the block $\leaf(V_i)$; see Figure~\ref{fig:elbow-eg}. 
An elbow is \textbf{maximal} if it is not contained in another elbow.
A leaf district is \textbf{elbow-free} if it does not contain any elbows.

\begin{figure}[h]
    \centering
    \includegraphics[scale=1.2]{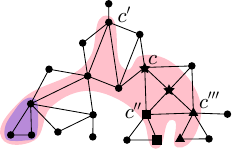}
    \caption{Example for the definitions of $\down(.)$ and elbow. District $V_i$ (pink) contains a leaf block (purple). Let $W$, $W'$, $W''$, and $W'''$ be the blocks whose parent cut vertices are $c$, $c'$, $c''$, and $c'''$, resp. 
    The set $\down(W') = V_i$ is not an elbow since it contains $\leaf(V_i)$.
    Vertices shown as squares (triangles) are in $\down(W'')$ ($\down(W''')$).
    The vertices in $\down(W)$ are shown as stars, squares and triangles.
    Sets $\down(W)$, $\down(W'')$ and $\down(W''')$ are elbows while only the first is maximal.}
    \label{fig:elbow-eg}
\end{figure}

A district map of $G$ is in \textbf{pseudo-canonical} form if every block $W\in B(G)$ satisfies one of the following three mutually exclusive conditions (see Figure~\ref{fig:pseudo-canonical} for examples):
\begin{enumerate}[label=(\roman*)]\itemsep 0pt
\item \label{type:singleton} all vertices in $W$ 
are in singleton nonleaf districts;

\item \label{type:leaf} all vertices of $W$, with the possible exception of the parent cut-vertex of $W$, are in the same leaf district.
Moreover, if $W$ is not a leaf block, then this district contains the leftmost grandchild block of $W$.

\item \label{type:sing+1} all vertices of $W$ are in nonleaf districts, whose vertices are all contained in $W$, but not all are singletons, and all ancestor (resp., descendant) blocks of $W$ are of type (i) (resp., type~\ref{type:leaf});
\end{enumerate}

We refer to the condition that a block satisfies as its \emph{type}.
Notice that \ref{type:sing+1} implies that blocks of type~\ref{type:singleton} (or blocks of types~\ref{type:singleton} and \ref{type:sing+1}) induce a connected subtree of $T$ containing the root.
The proof of Theorem~\ref{thm:general-graphs} is the combination of Lemmas~\ref{lem:general1} and~\ref{lem:general2}.

\begin{lemma}\label{lem:general1}
Let $G$ be a connected graph with $n$ vertices and let $1\leq k\le n$.
Every \contractible\ $k$-district map can be taken into
pseudo-canonical form by a sequence of $O(kn)$ switches.
\end{lemma}
\begin{proof}
Let $\Pi$ be a \contractible\ $k$-district map. Algorithm~\ref{algo:1conn} (below) transforms $\Pi$ into pseudo-canonical form in three phases; refer to Figure~\ref{fig:elbows}. 
Each phase processes all blocks in $B(G)$ in DFS order of the block tree $T$.
Phase~1 eliminates elbows. Phase~2 \contract s leaf districts such that they are each confined to their leaf blocks.
Phase~3 \contract s all nonleaf districts to singletons (or possibly turns some nonleaf districts into leaf districts). We continue with the details.

\begin{figure}[hp]
\begin{center}
\includegraphics[width=0.99\textwidth]{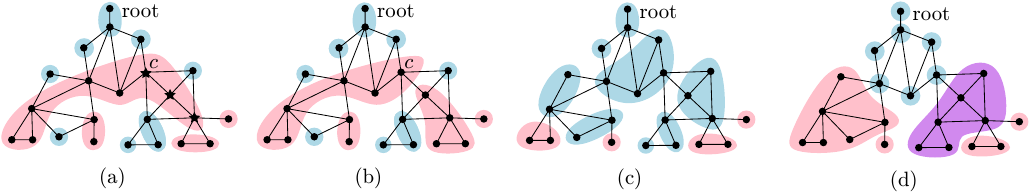}
\caption{(a) A 12-district map of a graph. The four leaf districts are red, eight nonleaf districts are blue; $c$ is the highest cut vertex in an elbow whose vertices are shown as stars. (b), (c), and (d) show the result of Phases 1, 2, and 3 of Algorithm \ref{algo:1conn}, respectively. 
In Phase~3, a nonleaf district becomes a leaf district (shaded purple).}
\label{fig:elbows}
\end{center}
\end{figure}

In lines 3 and 7, Algorithm~\ref{algo:1conn} \contract s $\down(W)$ and $W\cap V_i$, resp., into a singleton $\{c\}$. We describe these subroutines here. In both cases, we invoke Lemma~\ref{lem:con2} for a district map $\Pi'$ on a subgraph $G'$ of $G$; where $\Pi'$ is the restriction of $\Pi$ to $G'$. In the first case, $G'$ is induced by $c$ and its descendants. Note that $\down(W)$ is a district in $\Pi'$, and it is \contractible\ by Lemma~\ref{lem:con2} since $\down(W)$ is an elbow in $G$ and contains no leaf districts. In the second case, $G'$ is obtained from $G$ by deleting all descendants of $c$. Now $W\cap V_i$ is a district in $\Pi'$ since $V_i$ is a leaf district in $\Pi$, $W$ is the highest nonleaf block (in DFS order) that intersects $V_i$, and $V_i$ does not contain elbows by invariant (I1) below. By Lemma~\ref{lem:con2}, $W\cap V_i$ is \contractible\  as it lies in a single block $W$. In both cases, Lemma~\ref{lem:con2} yields a sequence of switches that \contract\ $\down(W)$ and $W\cap V_i$, resp., to $\{c\}$.

In lines 13 and 16, Algorithm~\ref{algo:1conn} \contract s a district $V_i$ with $c'\in V_i$ to $\{c'\}$. In both cases, $V_i$ is shrinkable to $c'$ by Lemma~\ref{lem:con2}, and the proof of Lemma~\ref{lem:con2} provides an algorithm that successively switches vertices in $V_i\setminus \{c'\}$ to other districts arbitrarily. 
However, this process might introduce a new elbow.
Here, we specify a particular a sequence of switches to ensure that no new elbows are created.
While $V_i$ is not a singleton, identify a noncut-vertex $v$ of $V_i$ adjacent to a vertex $w$ in a nonleaf district $V_j$.
(For example, see Figure~\ref{fig:phase3}(c)-(d) where $W_1$ has the role of $W'$.)
If no such vertex exists, choose $v$ adjacent to a vertex $w$ in a leaf district $V_j$ that intersects the leftmost grandchild block of $W'$.
(For example, see Figure~\ref{fig:phase3}(e)-(f) where $W_1$ has the role of $W'$.)
Let $u$ be a neighbor of $v$ in $V_i$, and apply \sw$(u,v,w)$.

\begin{algorithm}[H]
\caption{Pseudo-Canonical Algorithm for Connected Graphs}\label{algo:1conn}
\begin{algorithmic}[1]
\Procedure{Pseudo-Canonical}{$G,k,\Pi$}
\For{every nonroot block $W\in B(G)$ in DFS order of $T$}
    \If{$\down(W)$ is an elbow}
    {let $c\in C(G)$ be $W$'s parent, \contract\ $\down(W)$ to $\{c\}$.}
    \EndIf
\EndFor
\For{every nonleaf block $W\in B(G)$ in DFS order of $T$}
    \If{$W$ intersects a leaf district,}{
    \For{each leaf district $V_i$ that intersects $W$}
 \State \contract\ $W\cap V_i$ onto the cut-vertex $c$ of $W$ 
  in the descending path of $T$ to $\leaf(V_i)$;
  \State apply an additional switch to contract $V_i$ out of the block $W$.
    \EndFor}
    \EndIf
\EndFor
\For{every block $W\in B(G)$ in DFS order of $T$}
    \While{$W$ satisfies neither~\ref{type:singleton} nor~\ref{type:leaf}, and a grandchild $W'$ of $W$ is not of type~\ref{type:leaf}}
        \If{$W$ is the root of T and $W$ is contained in a single district,}
            \State Let $c'$ be a noncut-vertex of $W$, and let $V_i$ be the district containing $c'$;
            \State \Contract\ $V_i$ to $\{c'\}$.
        \Else
            \State Let $c'$ be the parent cut-vertex of $W'$, and let $V_i$ be the district containing $c'$;
            \State \Contract\ $V_i$ to $\{c'\}$.
            \If{$W'$ is still not of type~\ref{type:leaf} and $W$ is not of type~\ref{type:singleton},}
                \State Use Lemma~\ref{lem:pushing} with $P=(c')$ to move a district from $W'$ to $W$.
        \EndIf
        \EndIf
    \EndWhile
    \If{$W$ still satisfies neither~\ref{type:singleton} nor~\ref{type:leaf},}
        \State{\Contract\ the district containing the parent cut-vertex $c$ of $W$ to $\{c\}$.}
    \EndIf
\EndFor
\EndProcedure
\end{algorithmic}
\end{algorithm}

\smallskip\noindent\textbf{Analysis of Algorithm~2.}
Note that maximal elbows are pairwise disjoint, and every block intersects at most one maximal elbow (by the definition of $\down(W)$). 

Phase~1 (lines 2-3) iterates over all nonroot blocks. In the course of Phase~1, we maintain the invariant that if $W$ has been processed, then $\down(W)$ is not an elbow. When the for-loop reaches a block $W$ where $\down(W)$ is an elbow, then it is a maximal elbow due to the DFS traversal of $T$, and $\down(W)$ is \contracted\ to a cut vertex $c$, and produces $\down(W)=\{c\}$, which is not an elbow. We also show that this does not create any new elbows. Indeed, if a switch \contract s $\down(W)$ out of a cut vertex $c'$, then $c'$ is a descendant of $c$, and some district $V_j$ that intersects a child block $W'$ of $c'$ expands into $c'$. At this time, $c'$ becomes the highest vertex of $V_j$, and so $\down(W')$ contains $\leaf(V_j)$ if $V_j$ is a leaf district (hence $\down(W')$ cannot be an elbow). Thus, we conclude that Phase~1 successively eliminates all elbows and does not create any new elbow. Since the maximal elbows are pairwise disjoint, the sum of their cardinalities is at most $n$, and they can be \contracted\ with $O(n)$ switches. In Phases~2-3, we maintain invariant (I1): There are no elbows in the district map.

Phase~2 (lines 4-8) is a for-loop over all nonleaf blocks. In the course of Phase~2, we maintain the invariant that if $W$ has been processed, then $W$ is disjoint from leaf districts. When the for-loop reaches a block $W$ that intersects a leaf district $V_i$, then $V_i$ has no elbows by invariant (I1), and the ancestors of $W$ are disjoint from $V_i$ (because we visit blocks in DFS order). 
Consequently $V_i\cap W$ is \contractible\ to
the child of $W$ that leads to the leaf block $\leaf(V_i)$. 
For each leaf district $V_i$, Phase~2 uses $O(n)$ switches to \contract\ $V_i$, and $O(kn)$ switches overall. No elbows are created since leaf districts are never expanded to a block they do not already intersect (with the possible exception of the parent cut vertex of a block). In Phase~3, we maintain invariant (I2): If a leaf district intersects a block, then such block is of type~\ref{type:leaf}.

Phase~3 (lines 9-20) is a for-loop over all blocks $W\in B(G)$; see 
Figure~\ref{fig:phase3} for an example of the execution of this phase.
In the course of Phase~3, we maintain the invariant that if $W$ has been processed, it satisfies condition (i), (ii), or (iii) in the definition of pseudo-canonical forms. 
Indeed, for every block $W$, the switch operations modify only $W$ or its descendants. 
The fact that we are processing $W$ means that its grandparent is of type~\ref{type:singleton} when we begin processing $W$. Then, $W$ intersects more than one district and we can \contract\ $V_i$ to $\{c'\}$ in line 16 without expanding any districts not contained in $W$ and in ancestors of $W$.
This already implies that (I1) is maintained. Furthermore, if $W$ satisfies conditions~\ref{type:singleton} or \ref{type:leaf},
then the districts in $W$ remain unchanged. 
Otherwise, the while-loop (lines 10-18) ensures that every district that intersects $W$ is contained in $W$. 
In each iteration of the while-loop, $V_i$ is a nonleaf district by (I2), and $V_i$ is contained in the union of $W$ and its descendants.
The switches in lines 13 and 16 do not decrease the number of districts in $W$.
The preference of expansions to \contract\ $V_i$ to $\{c'\}$ ensures that
(I2) is maintained for leaf districts.
Indeed, such a switch may expand a leaf district if there is no other option:
in this case $V_i$ contains an entire block $W^*$, which is a descendant of $W$ and whose grandchildren are of type~\ref{type:leaf}; after \contract ing $V_i$ to the parent cut-vertex of $W^*$ expanding a leaf district, $W^*$ becomes of type~\ref{type:leaf}.
Using Lemma~\ref{lem:pushing} in line 18 ensures that, eventually, $W$ is of type~\ref{type:singleton} or \ref{type:leaf}, or all its grandchildren are of type~\ref{type:leaf}. Finally, when the while loop terminates, lines 19-20 ensure that the parent cut vertex of $W$ is a singleton, and so all ancestors of $W$ comprise singletons. 
In Phase~3, $O(n)$ switches \contract\ each district, amounting to $O(kn)$ switches overall.
%

We have shown that Algorithm~\ref{algo:1conn} takes any input district map $\Pi$ into pseudo-canonical form. The three phases jointly use $O(kn)$ switches, as claimed. 
\end{proof}

\begin{figure}[hp]
\begin{center}
\includegraphics[width=0.85\textwidth]{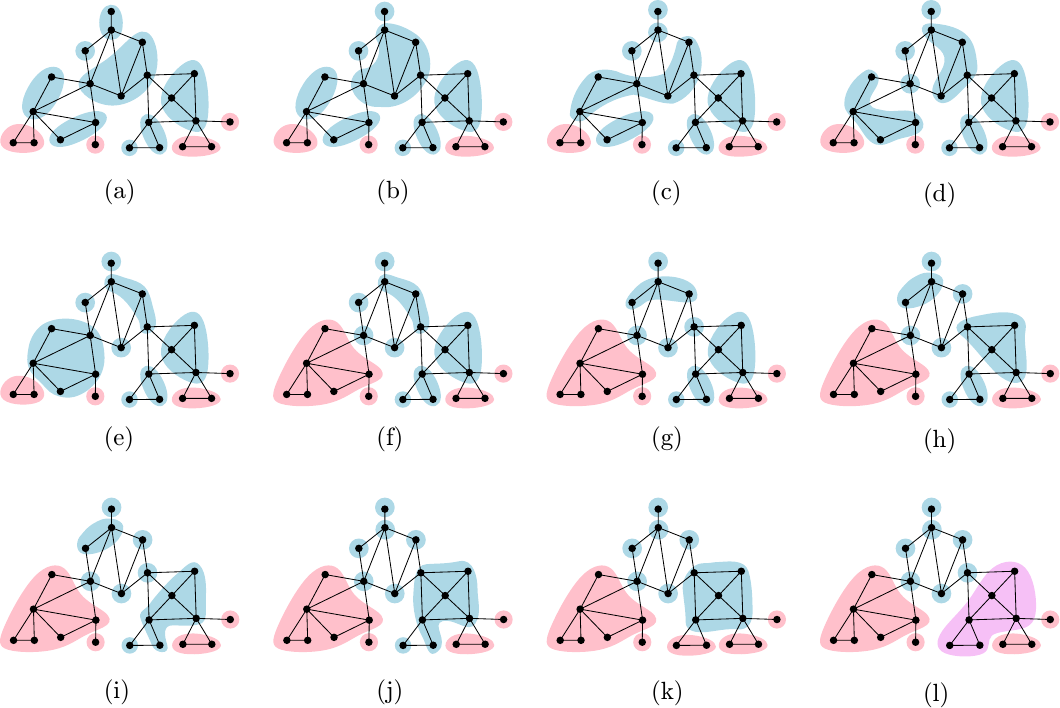}
\caption{Breakdown of the example of Phase 3 from  Figure~\ref{fig:elbows}(c) to Figure~\ref{fig:elbows}(d). 
The condition in line~10 is satisfied in (a) and (b) where $W$ is the root block, (c), (e), (f), (h) and (j) where $W$ is the grandchild of the root block.
Only in (a) the conditions in line 11 are satisfied.
Lemma~\ref{lem:pushing} (line 18) is applied in (d), (g), and (i).
While \contract ing a district $V_i$ in Figure~\ref{fig:phase3}(j), a nonleaf district becomes a leaf district in Figure~\ref{fig:phase3}(k). Continuing the \contract ing, causes $W_2$ to become type~\ref{type:leaf} in Figure~\ref{fig:phase3}(l).}
\label{fig:phase3}
\end{center}
\end{figure}

We now introduce a method to transform a pseudo-canonical $k$-district map into another.

\begin{lemma}\label{lem:general2}
Let $G$ be a connected graph with $n$ vertices and let $k\le n$ be a positive integer.
For any two pseudo-canonical $k$-district maps, $\Pi_1$ and $\Pi_2$, there is
a sequence of $O(kn)$ switches that take $\Pi_1$ to $\Pi_2$.
\end{lemma}
\begin{proof}
Our proof is constructive: for a given district map $\Pi$ in pseudo-canonical form, we assign every leaf district to the unique leaf block it intersects, and assign every nonleaf district to the highest (closest to the root) block in $T$ it is contained in. 
For every block $W\in B(G)$, let $d_\Pi(W)$ be the number of districts assigned to $W$ in $\Pi$. Notice that $\sum_{W\in B(G)}d_\Pi(W)=k$.

First, we explain how to transform $\Pi_1$ into an intermediate pseudo-canonical district map $\Pi_m$ so that $d_{\Pi_m}(W)=d_{\Pi_2}(W)$ for every block $W\in B(G)$.

Suppose that $d_{\Pi_1}(W)\neq d_{\Pi_2}(W)$ for some block $W\in B(G)$ (otherwise, we can trivially set $\Pi_m=\Pi_1$). 
Since $\sum_{W\in B(G)}d_{\Pi_1}(W)=\sum_{W\in B(G)}d_{\Pi_2}(W)=k$, 
there exist blocks $W_1,W_2\in B(G)$ such that $d_{\Pi_1}(W_1)< d_{\Pi_2}(W_1)$ and $d_{\Pi_1}(W_2)>d_{\Pi_2}(W_2)$.

\begin{claim}\label{cl:star} 
Let $W_1$ (resp., $W_2$) be a highest (resp., lowest) block such that $d_{\Pi_1}(W_1)< d_{\Pi_2}(W_1)$ (resp., $d_{\Pi_1}(W_2)>d_{\Pi_2}(W_2)$), then all ancestor blocks of $W_1$ and $W_2$ are of type~\ref{type:singleton} in $\Pi_1$, and all descendant blocks of $W_1$ and $W_2$ are of type~\ref{type:leaf} in $\Pi_1$.
\end{claim}
\begin{proof}
Notice that if a block is of type~\ref{type:singleton} (resp., type~\ref{type:leaf}) then it has been assigned with the maximum (resp., minimum) number of districts that it can possibly be assigned to. Then, $W_1$ cannot be of type~\ref{type:singleton} in $\Pi_1$ and it cannot be of type~\ref{type:leaf} in $\Pi_2$.
By the definition of pseudo-canonical forms, all descendant (resp., ancestor) blocks of $W_1$ are of type~\ref{type:leaf} (resp., type~\ref{type:singleton}) in $\Pi_1$ (resp., $\Pi_2$). By the choice of $W_1$, all ancestor blocks of $W_1$ are of type~\ref{type:singleton} in $\Pi_1$. 
An analogous argument proves the claim for $W_2$.
\end{proof}

Next, we construct an intermediate district map $\Pi_m$ by successively reducing the difference in the $d$ functions. While $d_{\Pi_1}\neq d_{\Pi_2}$, we transform $\Pi_1$ into another district map $\Pi_1'$ in pseudo-canonical form such that 
\begin{equation}\label{eq:potential}
\sum_{W\in B(G)}|d_{\Pi_1'}(W)-d_{\Pi_2}(W)| <
\sum_{W\in B(G)}|d_{\Pi_1}(W)-d_{\Pi_2}(W)|.
\end{equation}
Let $W_1$ and $W_2$ be blocks chosen as in Claim~\ref{cl:star}, let $c_1$ and $c_2$ be their respective parent cut-vertices, and let $P$ be a shortest path between $c_1$ and $c_2$. By Claim~\ref{cl:star}, all blocks along $P$ are of type~\ref{type:singleton} in $\Pi_1$, and so every vertex in $P$ is in a singleton district. 
Applying Lemma~\ref{lem:pushing} to $\Pi_1$, we can move a district from $W_2$ to $W_1$ using $O(|W_1|+|P|)\leq O(n)$ switches. 
We need to make sure that the new map is also in pseudo- canonical form.
If $d_{\Pi_1}(W_1)=0$ (i.e., $W_1$ is of type~\ref{type:leaf}, but not a leaf block) \contract\ the leaf district out of $W_1$ by expanding the new nonleaf district that has moved into $W_1$.
If $W_2$ consists of a single (nonleaf) district, \contract\ it onto $\{c_2\}$ while expanding the leaf district of its leftmost grandchild $W_2'$.
The number of districts assigned to a block changes only in $W_1$, $W_2$, and (possibly) $W_2'$. The procedure described above increases $d(W_1)$ by one, and decreases $d(W_2)$ (and possibly $d(W_2')$) by one, making the difference smaller as claimed. 
The type of $W_1$ (resp., $W_2$) becomes \ref{type:singleton} or \ref{type:sing+1} (resp., \ref{type:sing+1} or \ref{type:leaf}) and, by Claim~\ref{cl:star}, $\Pi_1'$ is in pseudo-canonical form.

In summary: while $d_{\Pi_1}\neq d_{\Pi_2}$, we repeat the above procedure. When the while loop ends, we find a pseudo-canonical district map $\Pi_m$ such that $\sum_{W\in B(G)}|d_{\Pi_m}(W)-d_{\Pi_2}(W)|=0$ (and thus $d_{\Pi_m}=d_{\Pi_2}$).
Initially $\sum_{W\in B(G)}|d_{\Pi_1}(W)-d_{\Pi_2}(W)|\leq 2k$ and each step decreases the difference by at least one, and so at most $2k$ iterations will be needed. Since each iteration takes $O(n)$ switches, this process uses $O(nk)$ switches overall.

In order to complete the proof of Lemma~\ref{lem:general2}, we need to show how to reconfigure $\Pi_m$ to $\Pi_2$. Recall that both district maps are in pseudo-canonical form and they satisfy $d_{\Pi_m}=d_{\Pi_2}$. Further, if a district map is in pseudo-canonical form, each block is of one of three possible types. We claim that every block of $G$ is of the same type in both $\Pi_m$ and $\Pi_2$. 
For ease of notation, we assume $\Pi_1=\Pi_m$. If $W$ is the root and $d_{\Pi_i}(W)=|W|$, $i\in\{1,2\}$, or $W$ is a nonroot block and $d_{\Pi_i}(W)=|W|-1$, then $W$ is of type \ref{type:singleton}.
If $W$ is a leaf block and $d_{\Pi_i}(W)=1$, or $W$ is a nonleaf block and $d_{\Pi_i}(W)=0$ then $W$ is of type~\ref{type:leaf}. 
Else, $W$ is of type~\ref{type:sing+1}.
This implies that every block of type~\ref{type:singleton} consists of singletons; and the union of blocks of type~\ref{type:leaf} are partitioned identically into leaf districts in both $\Pi_1$ and $\Pi_2$ since, by definition of type~\ref{type:leaf}, the leaf district that intersects the block must contain the leftmost grandchild block, and by the fact that there are no elbows. 
Thus, no switches are required in these blocks.
Blocks of type~\ref{type:sing+1} each contain the same number of districts in both $\Pi_1$ and $\Pi_2$. These blocks are pairwise disjoint by definition and all districts that intersect such a block is entirely contained in that block. Applying Algorithm~\ref{algo:2conn} to each block $W$ of type~\ref{type:sing+1}, both $\Pi_1$ and $\Pi_2$ transform to the same district map in $O(k|W|)$ switches s by Theorem~\ref{thm:2conn-alg}. Overall, this takes $O(kn)$ switches, completing the proof of Lemma~\ref{lem:general2}.
\end{proof}

\subsection{Characterization of Connected Switch Graphs}
\label{ssec:char}

Using Lemmas~\ref{lem:con1}--\ref{lem:invariant} and Theorem~\ref{thm:general-graphs},
we can characterize the pairs $(G,k)$, of a connected graph $G$ and a positive integer $k$,
for which the switch graph $\Gamma_k(G)$ is connected\ShoLong{}{ (cf.~Theorem~\ref{thm:conn-test} below)}.

\begin{lemma}\label{lem:conn-char}
For a connected graph $G$ with $n$ vertices and an integer $1\leq k\le n$, the switch graph $\Gamma_k(G)$ is connected if and only if $k=1$ or every $k$-district map is \contractible\ (i.e., $\Gamma_k(G)=\Gamma_k'(G)$).
\end{lemma}

\global\def \pflemconnchar {
\begin{proof}
The case that $k=1$ is trivial, as $\Gamma_k(G)$ is a singleton. Assume $k\geq 2$ for the remainder of the proof.
If every $k$-district map is \contractible\ (i.e., $\Gamma_k(G)=\Gamma_k'(G)$), then $\Gamma_k'(G)$ is connected  by Theorem~\ref{thm:general-graphs}, and so $\Gamma_k(G)$ is connected. If some $k$-district maps are \contractible\ and some are \incontractible, then $\Gamma_k(G)$ is disconnected, since there is no edge between the set of \contractible\ and \incontractible\ district maps by Lemma~\ref{lem:invariant}.

Finally, assume that every $k$-district map is \incontractible\ in $G$. We show that $\Gamma_k(G)$ is disconnected. Let $\Pi_1$ be an arbitrary $k$-district map. By Lemmas~\ref{lem:con1}--\ref{lem:con2}, some district $V_i\in \Pi_1$ contains two leaf blocks of the block graph, say $W_a,W_b\in B(G)$, with parent cut vertices $c_a,c_b\in C(G)$ (possibly $c_a=c_b)$. Since $G$ is connected and $k\geq 2$, there exists a district $V_j$ adjacent to $V_i$. We construct a $k$-district map $\Pi_2$ from $\Pi_1$ by replacing $V_i$ and $V_j$ with $V_i':=W_a\setminus \{c_a\}$ and $V_j':=(V_i\cup V_j)\setminus V_i'$. Importantly, none of the districts in $\Pi_2$ contain both $W_a$ and $W_b$; and by Lemma~\ref{lem:con1}, every sequence of switch operations transforms $V_i$ to a district that contains both $W_a$ and $W_b$.
Thus $\Gamma_k(G)$ does not contain any path between $\Pi_1$ and $\Pi_2$, as required.
\end{proof}
}
\ShoLong{Proof is deferred to the Appendix.}{\pflemconnchar

}
Lemma~\ref{lem:con1} allows us to efficiently check whether a connected graph $G$ admits an \incontractible\ $k$-district map. Let $G$ be connected but not biconnected. For two leaf blocks $W_1,W_2\in B(G)$, let $P(W_1,W_2)$ denote the \emph{union} of $W_1$, $W_2$, and the set of vertices along a shortest path in $G$ between $W_1$ and $W_2$. Let
%
$M=\min\{|P(W_1,W_2)|: W_1,W_2\in B(G) \mbox{ \rm leaf blocks}\}$.
%

\begin{lemma}\label{lem:conn-test}
Let $G$ be a connected graph with $n$ vertices that is connected but not biconnected, and let $k\geq 2$.
Every $k$-district map in $G$ is contractible if and only if 
$n-k+2\leq M$.
\end{lemma}
\begin{proof}
If $n-k+2\leq M$, then every district in a $k$-district map contains at most $n-(k-1)=n-k+1\leq M-1$ vertices. By the definition of $M$, none of the districts can contain two leaf blocks, therefore all $k$ districts are contractible by Lemma~\ref{lem:con2}. 

If $M\leq n-k+1$, then we construct a $k$-district map for $G$ in which one of the districts is incontractible.
Let $\widehat{V}\subset V(G)$ be a set of $M$ vertices that includes two leaf blocks in $B(G)$ and a shortest path between them. By partitioning $V(G)\setminus \widehat{V}$ into singletons, we obtain a $\widehat{k}$-district map $\widehat{\Pi}$, where $\widehat{k}=n-M+1\geq k$, and $\widehat{V}\in \widehat{\Pi}$. Successively merge pairs of adjacent districts until the number of districts drops to $k$ (recall that $G$ is connected, so some pair of districts are always adjacent). We obtain a $k$-district map $\Pi$, in which one of the districts contains $\widehat{V}$, and is incontractible by Lemma~\ref{lem:con1}, as required.
\end{proof}

\begin{lemma}\label{lem:BFS}
We can compute the value $M$ in $O(n+m)$ time, where $n=|V(G)|$ and $m=|E(G)|$.
\end{lemma}
\begin{proof}
Given a connected graph $G=(V,E)$, first compute the block tree, and modify $G$ as follows: replace each leaf block by a path with the same number of vertices, such that one endpoint is the original cut vertex (and hence the other endpoint is a leaf), and denote by $G'$ the resulting graph. Then we run a modified multi-source BFS on $G'$, starting from the leaves. The algorithm assigns two labels to every vertex $v\in V(G')$, the \emph{level} $\ell(v)$ and a \emph{cluster} $c(v)$. Initially, each leaf $v\in V(G')$ is assigned level $\ell(v)=0$ and clusters $c(v)=v$. When the BFS visits a new vertex $v$ along an edge $uv$, it sets $\ell(v):=\ell(u)+1$ and $c(v):=c(u)$. Clearly, $\ell(v)$ is the distance from $v$ to the closest leaf in $G'$, and $c(v)$ is one such leaf. After the BFS termination, our algorithm finds an edge $uv\in E(G')$ such that $c(u)\neq c(v)$ and $\ell(u)+\ell(v)$ is minimal, and returns $\ell(u)+\ell(v)+2$.

The modified BFS runs in $O(n+m)$ time, and a desired edge $uv$ can be found in $O(m)$ additional time, so the overall running time is $O(n+m)$. It remains to prove that $M=\ell(u)+\ell(v)+2$. Note that $u$ and $v$ are at distance $\ell(u)$ and $\ell(v)$, resp., from the leaves $c(u)$ and $c(v)$. The cluster of $c(u)$ (resp., $c(v)$) contains a shortest path from $u$ to $c(u)$ (resp., from $v$ to $c(v)$), and so these shortest paths are disjoint. The concatenation of the two shortest paths is a shortest path $P'$ between the leaves $c(u)$ and $c(v)$, and it has $\ell(u)+\ell(v)+2$ vertices. The path $P'$ contains the chains incident to $u$ and $v$ in $G'$. By the definition of $G'$, these chains correspond to leaf blocks $W_1$ and $W_2$ of the same size in $G$. Consequently, $\ell(u)+\ell(v)+2 = P(W_1,W_2)$
Therefore, $M\leq \ell(u)+\ell(v)+2$.

Conversely, assume that $M=P(W_1,W_2)$ for some leaf blocks $W_1,W_2\in B(G)$. These leaf blocks correspond to chains ending in two leaves, say $W'_1$ and $W'_2$, in $G'$. By construction, the distance between $W'_1$ and $W'_2$ is $d_{G'}(W'_1,W'_2)=M-1$.  Let $P'$ be a shortest path between $W'_1$ and $W'_2$. We claim that for every vertex $v'$ in $P'$, $\ell(v')$ is the minimum distance to $\{W'_1,W'_2\}$, i.e., $\ell(v')=\min\{d_{G'}(v,W_1),d_{G'}(v,W_2)\}$. 
Suppose, to the contrary, that there is a vertex $v'$ in $P'$ for which  $\ell(v)\neq \min\{d_{G'}(v',W'_1),d_{G'}(v',W'_2)\}$. Since $\ell(v')$ is the minimum distance to some leaf in $G'$, we have $\ell(v')=d_{G'}(v,W'_3)$ for a leaf $W'_3$, and $\ell(v')<\min\{d_{G'}(v',W'_1),d_{G'}(v',W'_2)\}$. As $v'$ is in the path $P'$, 
$d_{G'}(v',W'_1)+d_{G'}(v',W'_2)=d_{G'}(W'_1,W'_2)=M-1$. By the triangle inequality, $d_{G'}(W'_1,W'_3)$ or $d_{G'}(W'_2,W'_3)$ is less than $d_{G'}(W'_1,W'_2)$, 
contradicting the minimality of $P(W_1,W_2)$. Now $P'$ contains two consecutive vertices, say $u^*$ and $v^*$, such that $\ell(u^*)=d_{G'}(u^*,W'_1)$, $\ell(v^*)=d_{G'}(v^*,W'_2)$, and $c(u^*)\neq c(v^*)$. The sum of their distances to the two endpoints of $P'$ is $\ell(u^*)+\ell(v^*)=(M-1)-1=M-2$, hence $M=\ell(u^*)+\ell(v^*)+2$. 
Then, $\ell(u)+\ell(v)+2\leq M$, as required. 
%
\end{proof}

The combination of Theorem~\ref{thm:2conn-alg} and Lemmas~\ref{lem:conn-char}--\ref{lem:BFS}
yields the following result.
\begin{theorem}\label{thm:conn-test}
For a connected graph $G$ with $n$ vertices and a positive integer $k\le n$, the switch graph $\Gamma_k(G)$ is connected if and only if $G$ is biconnected,  $n-k+2\leq M$, or $k=1$, which can be tested in 
$O(n+m)$ time, where $m=|E(G)|$.
\end{theorem}

\section{PSPACE-Completeness for Connectedness}
\label{ssec:hardness-connect}
\setcounter{tocdepth}{4}
\setcounter{secnumdepth}{4}

In the \textbf{connectedness problem}, we are given a graph $G$, and two $k$-district maps, $\Pi_A$ and $\Pi_B$, for some integer $1\leq k \le n$, and ask whether $\Pi_A$ and $\Pi_B$ are in the same component of the switch graph $\Gamma_k(G)$. In this section, we show that this problem is PSPACE-complete.
Further, we show that the problem remains PSPACE-complete even if $(i)$ we restrict $G$ to be a planar graph of maximum degree 6, or $(ii)$ we restrict the number of districts to $k=2$. As an immediate consequence, we show that the diameter of a connected component of $\Gamma_k(G)$ may be as large as $2^{\Omega(n)}$ where $n = |V(G)|$.
\FullSODA{
\subsection{Membership in PSPACE}
Since NPSPACE = PSPACE, to place our problem in PSPACE it suffices to give a nondeterministic algorithm which solves it using polynomial space. A more in depth discussion can be found in~\cite{PSPACE-book}.

\begin{algorithm}[h]
\caption{to Decide whether $A$ and $B$ are in the same Component of  $\Gamma_k(G)$}\label{algo:pspaceconn}
\begin{algorithmic}[1]
\Procedure{Connected}{$G,k,A,B$}
\State Label the vertices of $G$ according to the current district map $\Pi$, which is initialized to $A$.
\State Initialize an $n \lceil \log n \rceil$ bit counter $X$ to zero, where $n = |V(G)|$.
\While {$X < n^n$}
    \State Nondeterministically guess an ordered triplet of vertices $(u, v, w)$ from $G$.
    \If{switch$_{\Pi}(u, v, w)$ is a valid switch operation}
        \State $\Pi:=$ switch$_{\Pi}(u, v, w)$.
        \If{$\Pi = B$}
            \Return True
        \Else{}
                \State $X := X + 1$
        \EndIf
    \Else{}
        Continue
    \EndIf
\EndWhile
\State \Return False
\EndProcedure
\end{algorithmic}
\end{algorithm}

Consider Algorithm~\ref{algo:pspaceconn}. First note that the number of possible $k$-district maps on $G$ is bounded above by the number of partitions on an $n$ element set, which in turn is bounded above by $n^n$. So if $A$ and $B$ are in the same component of $\Gamma_k(G)$ then there exists a path from $A$ to $B$ whose length is upper bounded by $n^n$.
So the algorithm must find such a path before the counter reaches its terminal value.
Conversely, if $A$ and $B$ are not in the same component of $\Gamma_k(G)$ the algorithm will terminate once the counter reaches $n^n$.
The space required is polynomial since the only memory needed is to store at any given time the value of an $n \lceil \log n \rceil$ bit counter and a copy of $G$ with its vertices partitioned into districts.
}{Membership in PSPACE is justified by the fact that a nondeterministic machine can explore $\Gamma_k(G)$ storing one district map at a time, so we focus now on proving hardness.}

\subsection{PSPACE-Hardness for General Graphs with Many Districts}
\label{ssec:PSPACE-general}
We prove PSPACE-hardness by a reduction from the reconfiguration problem for Nondeterministic Constraint Logic (abbreviated NCL), which is known to be PSPACE-complete~\cite{PSPACE}.
In this problem, we are given an \emph{NCL graph}, which is a planar cubic graph where each edge is colored either blue or red, and each vertex is either an \emph{OR} vertex incident on 3 blue edges, or an \emph{AND} vertex incident on 2 red edges and 1 blue edge.
An NCL graph with an orientation assigned to its edges is considered \emph{satisfied} if all of its vertices are satisfied; an OR vertex is \emph{satisfied} when at least one edge is oriented towards it, and an AND vertex is \emph{satisfied} when both of its red edges are oriented towards it or its blue edge is oriented towards it.
In the \emph{NCL reconfiguration problem}, we are given an initial and a final orientation that both satisfy an NCL graph, and must decide whether one can be reconfigured into the other by flipping the orientation of one edge at a time in such a way that after each flip the NCL graph is satisfied.

Given an NCL graph $G_{NCL}$, we create a graph $G$ as follows.
First create OR and AND gadgets of 7 and 9 vertices, respectively. The adjacencies between the vertices are shown in Figure~\ref{fig:vertex gadgets}.
Given a vertex $v \in G_{NCL}$, we denote the corresponding gadget $F(v)$.
For each gadget, the labelled vertices in Figure~\ref{fig:vertex gadgets} are \emph{terminals}, and the unlabelled vertices that are adjacent to the leaves are called \emph{anchors}.

\begin{figure}[H]
	\centering
	\begin{subfigure}{.49\textwidth}
		\centering
		\includegraphics[width=.9\linewidth]{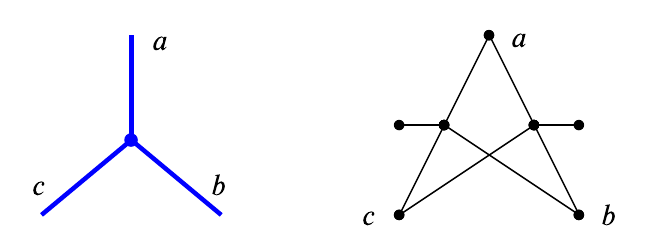}
		\label{fig:OR}
	\end{subfigure}
	\begin{subfigure}{.49\textwidth}
		\centering
		\includegraphics[width=.9\linewidth]{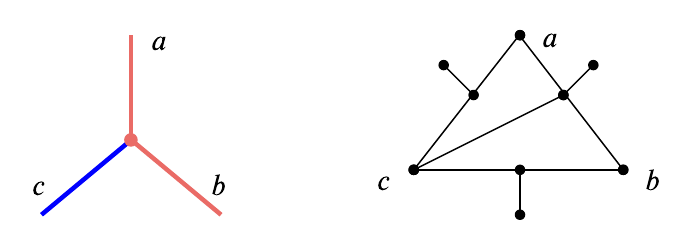}
		\label{fig:AND}
	\end{subfigure}
	\caption{Gadgets for OR and AND vertices (left and right, respectively)
	}
	\label{fig:vertex gadgets}
\end{figure}

For the AND gadget, we call the two degree-two terminals ($a$ and $b$) \emph{red terminals}, and the degree-three terminal ($c$) a \emph{blue terminal}.
Each edge of $G_{NCL}$ corresponds to a terminal in two gadgets, one for each vertex that edge is incident to (as shown by the labels in Figure~\ref{fig:vertex gadgets}), and thus we identify terminals of different gadgets that correspond to the same edge (Figure~\ref{fig:gadget gluing}).
This concludes the construction of $G$.

\begin{figure}[H]
	\centering
	\includegraphics[scale=0.6]{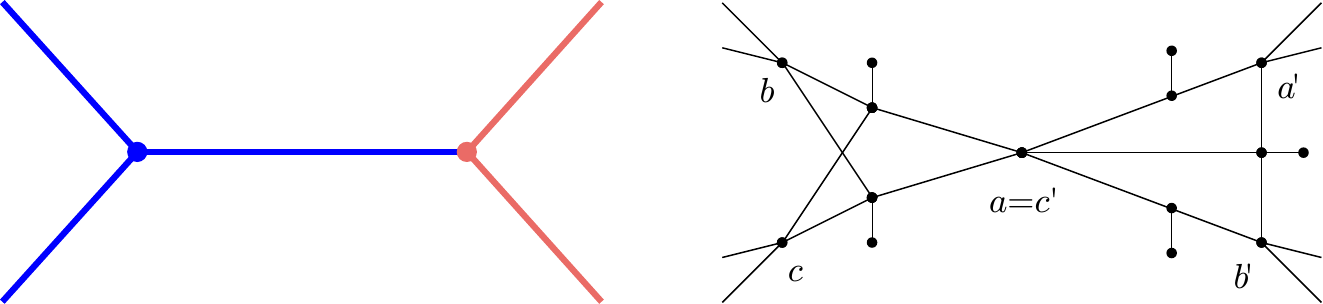}
	\caption{Two gadgets glued together along a shared terminal.
	}
	\label{fig:gadget gluing}
\end{figure}

It remains to construct an initial and a final district map for $G$ to simulate the initial and final orientations in $G_{NCL}$.
Given an orientation $O$ on $G_{NCL}$, construct a district map $\Pi_O$ for $G$ as follows: for every vertex $v \in G_{NCL}$,  we construct a district that will contain most vertices of the gadget $F(v)$. This district always contains all anchors and leaves of that gadget. In addition, it contains the terminal corresponding to every edge $e$ orientated towards $v$ in $O$ (see Figure~\ref{fig:orientation}).
Note that every district in the initial or the final configuration contains at least two leaves. By Lemma~\ref{lem:con1}, these leaves and their anchors cannot be moved to another district by any sequence of switch operations. Thus each district is tied to its respective gadget, that is,

\begin{enumerate}
    \item [$(\star)$] there is a one-to-one correspondence between the districts and the gadgets that remains invariant under 1-switch operations.
\end{enumerate}

\begin{figure}[H]
	\centering
	\includegraphics[scale=1]{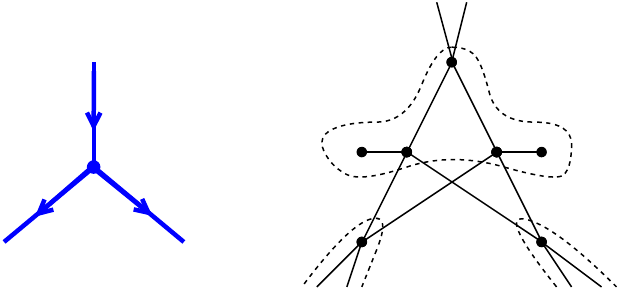}
	\caption{An orientation of an OR vertex and its associated district map
	}
	\label{fig:orientation}
\end{figure}

\begin{lemma}\label{lem:connectedness satisfies}
The district map $\Pi_O$ over $G$ is well defined if and only if the orientation $O$ on $G_{NCL}$ is satisfying.
\end{lemma}

\begin{proof}
%
Consider an AND vertex and its associated gadget. We know by property $(\star)$, that the district corresponding to this gadget must contain three anchor vertices and three leaves. In order for the district to be connected it must contain
(i) either the terminal associated with the blue edge ($c$ in Figure~\ref{fig:vertex gadgets}), or (ii) both terminals associated with red edges ($a$ and $b$ in Figure~\ref{fig:vertex gadgets}). This is equivalent to saying that the AND vertex is satisfied by $O$.

Similarly, for the gadget associated to an OR vertex, its corresponding district has two anchors and two unlabelled leaves. In order for these vertices to remain connected within the district, the district must contain any of the three terminals of the gadget. This is equivalent to saying that $O$ satisfies the OR vertex.
\end{proof}

\begin{lemma}[Flip-Switch Equivalence]\label{lem:switch to flip}
For every district map $\Pi_{O_1}$ on $G$ obtained from an orientation $O_1$ on $G_{NCL}$, every 1-switch operation on $\Pi_{O_1}$ yields a map $\Pi$ such that $\Pi = \Pi_{O_2}$ where $O_2$ is an orientation on $G_{NCL}$ that differs from $O_1$ by the orientation of a single edge.
Similarly, for an orientation $O_3$ on $G_{NCL}$ obtained from $O_1$ by flipping the orientation of a single edge, there is a 1-switch operation that takes $\Pi_{O_1}$ to $\Pi_{O_3}$.
\end{lemma}
\begin{proof}
Consider an edge $e$ in $G$ whose endpoints are in different districts in $\Pi_{O_1}$.
By construction of our gadgets and Lemma~\ref{lem:con1}, $e=rt$ for an anchor $r$ and a terminal $t$. Further, note that the district $\Pi_{O_1}(t)$ containing $t$ cannot absorb $r$ since the leaf adjacent to $r$ would disconnect from the rest of its district, violating the requirement that districts stay connected. Thus, the only 1-switch we can perform along $e$ is one in which the district $\Pi_{O_1}(r)$ containing $r$ expands to $t$. This is precisely the district map $\Pi_{O_2}$ where $O_2$ is the orientation we get by flipping the orientation of the edge associated with $t$ in $O_1$.

Conversely, let $e'$ be the only edge in $G_{NCL}$ whose orientation differs in $O_1$ and $O_3$. By construction, $e'$ corresponds to a terminal $v_{e'}$, which is adjacent to anchors in two distinct districts, and $\Pi_{O_1}$ and $\Pi_{O_3}$ differ only by the membership of $v_{e'}$. Hence $\Pi_{O_3}$ is obtained from $\Pi_{O_3}$ by performing a single 1-switch operation that moves $v_{e'}$ from one district to the other.
\end{proof}


\FullSODA{Now that we have equivalency between the two operations, we show that indeed the two problem instances are equivalent.
\begin{lemma}\label{lem:equivalency}
Given two orientations $O_1, O_2$ on the NCL graph $G_{NCL}$ with $k$ vertices, there exists a valid sequence of edge flips transforming $O_1$ to $O_2$ if and only if $\Pi_{O_1}, \Pi_{O_2}$ are in the same connected component of $\Gamma_{k}(G)$.
\end{lemma}
\begin{proof}
Say there exists a valid reconfiguration from $O_1$ to $O_2$ by edge flips. So there exists a sequence $O_{r_1}, \ldots, O_{r_k}$ of orientations on $G_{NCL}$ such that all orientations satisfy $G_{NCL}$, $O_{r_1} = O_1$, $O_{r_k} = O_2$, and $\forall i$ s.t. $1 \leq i < k$, $O_{r_i}$ differs from $O_{r_{i + 1}}$ by a single edge flip. By Lemma~\ref{lem:connectedness satisfies}, since $O_{r_i}$ satisfies $G_{NCL}$ for all $1 \leq i < k$, $\Pi_{O_{r_i}} \in \Gamma_{k}(G)$ for all $1 \leq i < k$. By  Lemma~\ref{lem:switch to flip}, since $O_{r_i}$ differs from $O_{r_{i + 1}}$ by a single edge flip for all $1 \leq i < k$, $\Pi_{O_{r_i}}$ is adjacent to $\Pi_{O_{r_{i+1}}}$ in $ \Gamma_{k}(G)$ for all $1 \leq i < k$. So then $\Pi_{O_{r_1}}, \ldots, \Pi_{O_{r_k}}$ describes a walk from $\Pi_{O_1}$ to $\Pi_{O_2}$ in $\Gamma_{k}(G)$, so the two district maps must be in the same component of $\Gamma_{k}(G)$. The other direction of the proof is essentially identical.
\end{proof}
}{
Lemma~\ref{lem:switch to flip} implies the following theorem:
}

\begin{theorem} \label{thm:general pspace hardness}
Given a graph $G$, and two $k$-district maps $\Pi_A$ and $\Pi_B$ on $G$, it is PSPACE-complete to determine whether $\Pi_A$ and $\Pi_B$ are in the same connected component of $\Gamma_k(G)$.
\end{theorem}
\FullSODA{
\begin{proof}
As shown above this problem lies in PSPACE. Since the edge flip reconfiguration problem on NCL graphs is PSPACE-hard, and our reduction clearly takes polynomial time, Lemma~\ref{lem:equivalency} shows that our problem is PSPACE-hard and thus PSPACE-complete.
\end{proof}
}{}

\subsection{PSPACE-Hardness for Planar Graphs}
\label{ssec:PSPACE-planar}
We now modify the reduction described in Section~\ref{ssec:PSPACE-general} so that it creates a planar graph $G$ for a given  NCL graph $G_{NCL}$. Recall that the NCL graph  $G_{NCL}$ is a planar cubic graph. Thus, to ensure that $G$ is planar, it suffices that each gadget admits a planar drawing with terminals in the outer face.

The gadget associated with an AND vertex already satisfies this condition. In this section, we construct a slightly more complicated gadget which behaves like an OR vertex and admits a planar drawing with all of its terminals on the outer face. Refer to Figure~\ref{fig:planar or}.
As before, the labeled vertices are the terminals for this gadget, and each terminal must be identified with the terminal of the neighboring gadget as previously described.
Apart from terminals, there are three leafs and three anchors and an copy of $K_3$ (i.e., a 3-cycle) in the middle of the gadget. The three anchors are adjacent to distinct vertices of the 3-cylce.

\begin{figure}[H]
	\centering
	\includegraphics[scale=0.4]{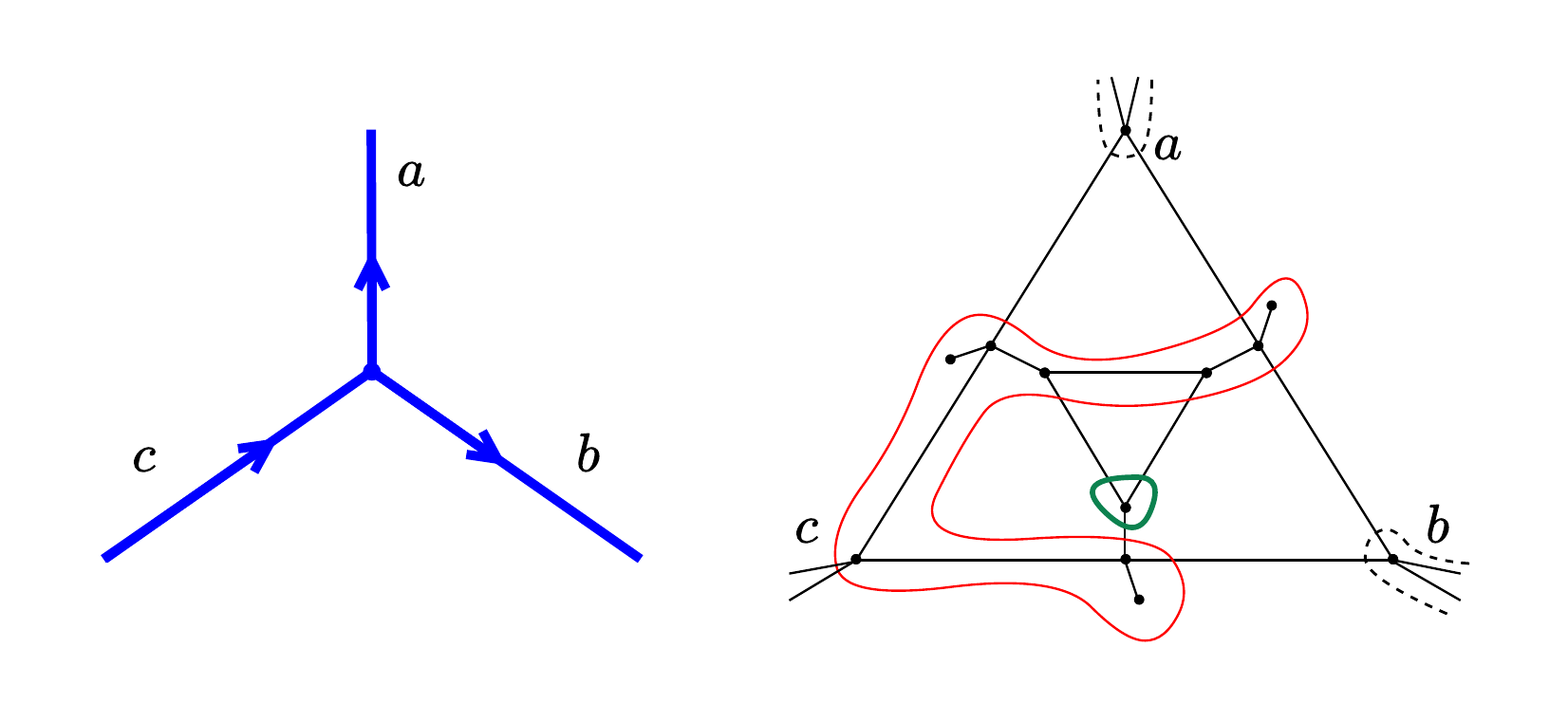}
	\caption{Modified planar OR gadget}
	\label{fig:planar or}
\end{figure}

Unlike the previous gadgets, this gadget comes equipped with two districts: one called the \emph{guard} that interacts with other gadgets (shown in thin red in Figure~\ref{fig:planar or}), defined as the district that  contains all three leafs and anchors of the gadget, and one called the \emph{prisoner} that is trapped inside of this gadget (shown in bold green in Figure~\ref{fig:planar or}), defined as a district that is not the guard and that consists of some vertices of the 3-cycle.
The correspondence between valid orientations of $G_{NCL}$, and district maps containing a guard and a unique prisoner district is defined as follows.
At least one vertex in the 3-cycle in the gadget is in the prisoner district.
We define only the prisoner district and let the guard district contain all remaining vertices of the gadget, except the terminals associated with edges of $G_{NCL}$ oriented away from the corresponding vertex (see Figure~\ref{fig:planar or}).
If the indegee of the OR vertex in $G_{NCL}$ is 2  or 3, the prisoner district can be any nonempty set of vertices in the 3-cycle. Otherwise, the indegree of the OR vertex is exactly one, and the guard district will contain exactly one of the three terminals. In that case, the prisoner district contains either of the two vertices of the 3-cylce that are within distance $2$ from the terminal vertex within the guard district. We now prove that, after any sequence of 1-switches, these properties of the guard and prisoner districts continue to hold, and they each induce a connected subgraph in $G$ if and only if they correspond to a satisfying orientation of $G_{NCL}$.

\begin{lemma} \label{lem:planar or}
The modified OR gadget behaves as the original OR gadget.
\end{lemma}
\begin{proof}
As was the case for the original OR gadget, we can see that a satisfying orientation of the NCL graph will create a connected district map for this new gadget and vice versa (the guard district will be connected if and only if it contains at least one terminal, and the prisoner district is always connected because it consists of a subgraph of the complete graph $K_3$).

Once again, Lemma~\ref{lem:con1} guarantees that the three leaves and the adjacent anchors of a guard district remain in the same district under any sequence of switch operations. Further, since the prisoner is on a 3-cycle which is only adjacent to cut vertices (anchors) in the guard district, which remain in the same district by Lemma~\ref{lem:con1}, the prisoner can never escape from this 3-cycle and the number of prisoners cannot change in a gadget. Thus, similarly to property ($\star$) in Section~\ref{ssec:PSPACE-general}, we can uniquely identify an OR gadget with the two districts it must always contain.

\begin{figure}[htbp]
	\centering
	\begin{subfigure}{.3\textwidth}
		\centering
		\includegraphics[width=.7\linewidth]{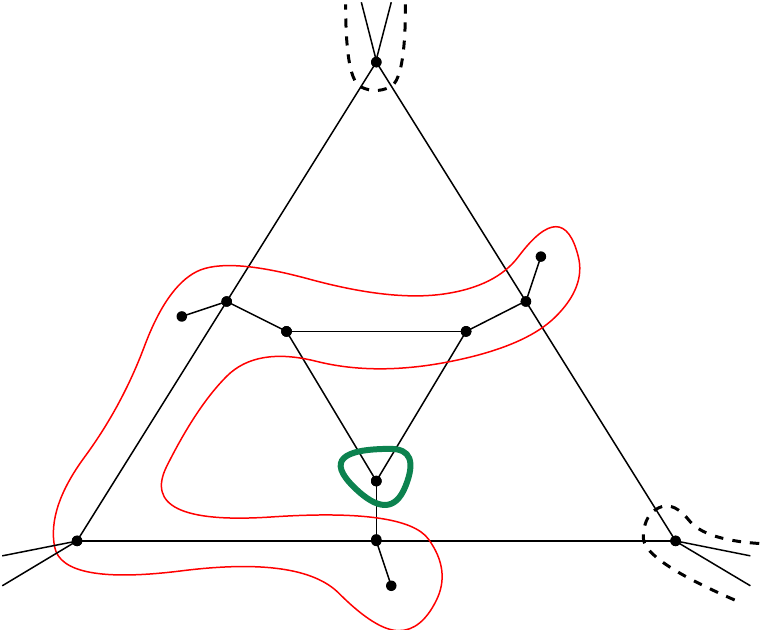}
		\label{fig:pt1}
	\end{subfigure}%
	\begin{subfigure}{.3\textwidth}
		\centering
		\includegraphics[width=.7\linewidth]{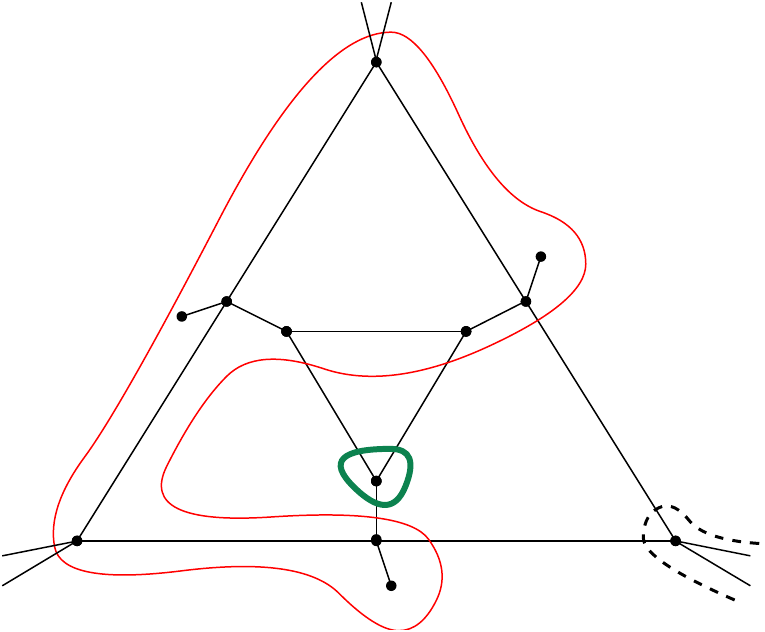}
		\label{fig:pt2}
	\end{subfigure}%
	\begin{subfigure}{.3\textwidth}
		\centering
		\includegraphics[width=.7\linewidth]{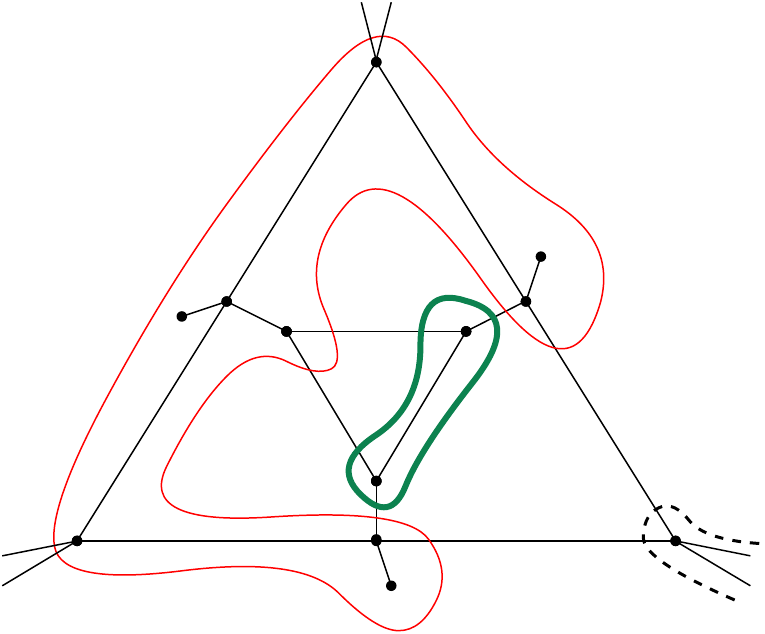}
		\label{fig:pt3}
	\end{subfigure}
	\begin{subfigure}{.3\textwidth}
		\centering
		\includegraphics[width=.7\linewidth]{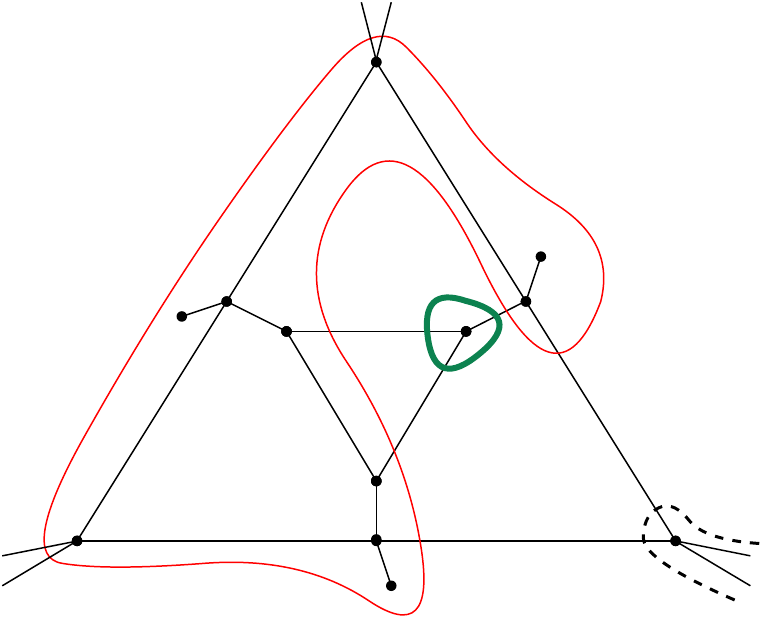}
		\label{fig:pt4}
	\end{subfigure}%
	\begin{subfigure}{.3\textwidth}
		\centering
		\includegraphics[width=.7\linewidth]{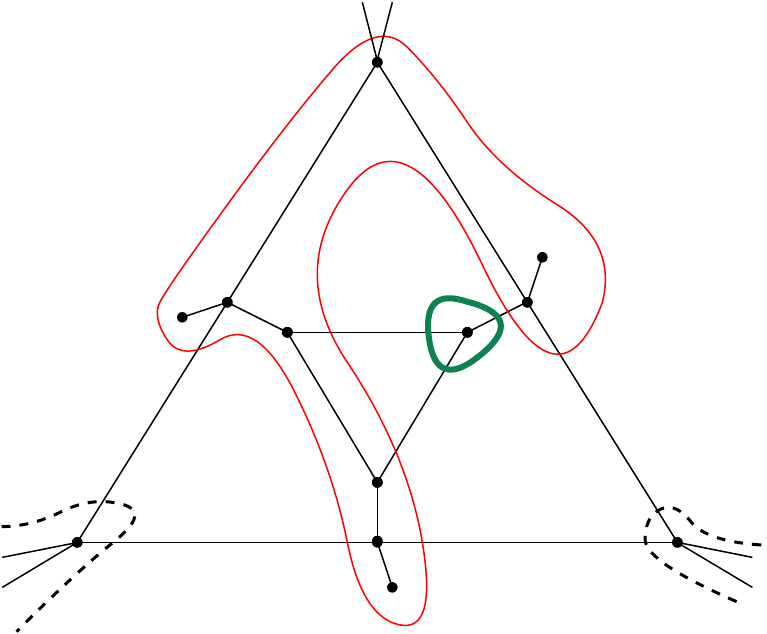}
		\label{fig:pt5}
	\end{subfigure}

	\caption{Transitioning from owning only the bottom left terminal to owning only the top terminal}
	\label{fig:transition}
\end{figure}

We have seen that the modified OR gadget is satisfied in a static state by the same conditions as the original OR gadget; it remains to show that there is a valid transition between any two states corresponding to valid orientations of edges in $G_{NCL}$ that differ by a flip.

As noted above, a guard district always contains all of its initial three anchors, and every terminal in its OR gadget is either in the guard district or adjacent to an anchor in the guard district. Therefore, a sequence of 1-switch operations can always transition to a state where the guards district  contains two or more terminals (two or more incoming blue edges). If the guard gadget contains two or three terminals, the prisoner district can move freely within the central triangle.
It follows that, from a state where the guard district contains two terminals (two incoming blue edges), the modified OR gadget can transition to a state where it contains only one of the original two (a single incoming blue edge).
%
%
As an illustration, Figure~\ref{fig:transition} shows how the gadget can transition between states where the guard district contains any one single terminal (a single blue edge directed inwards) through intermediate states where it contains two (two blue edges oriented inwards). Note that there are two different district maps that correspond to the same orientation when there is a single blue edge is oriented towards an OR vertex.

Conversely, since the three terminals of a modified OR gadget are incident to distinct other gadgets (as $G_{NCL}$ is a simple graph), a single switch can only add or remove one terminal to a guard district, hence only states that represent orientations differing by a single flip are adjacent in $\Gamma_k(G)$.
%
\end{proof}

\begin{theorem} \label{thm:planar pspace}
Given a planar graph $G$ of maximum degree 6, and two $k$-district maps $\Pi_A$ and $\Pi_B$ over $G$, it is PSPACE-complete to determine if $\Pi_A$ and $\Pi_B$ are in the same connected component of $\Gamma_k(G)$.
\end{theorem}

\subsection{PSPACE-Hardness for Two-District Maps}
\label{ssec:PSPACE-two}
In order to prove hardness in presence of only two districts, we modify the reduction described in Section~\ref{ssec:PSPACE-general} as follows. We start by subdividing every edge in the NCL graph $G_{NCL}$ and creating degree-two vertices which are satisfied so long as they have in-degree at least one. The addition of these vertices has no effect on the reconfiguration space; these extra vertices simply propagate signals from one vertex to another.
We can then assume that $G_{NCL}$ is bipartite: one partite set containing only degree-2 vertices, and the other partite set containing the original AND and OR vertices.

To build a gadget that simulates degree-two NCL vertices, simply take the original nonplanar OR gadget and delete one of its terminals. The resulting gadget has two terminals, both of which provide independent paths between the two anchor/leaf pairs in the gadget, so the district corresponding to this gadget will be connected if and only if it owns at least one of its terminals. An example of this gadget appears in the center of Figure~\ref{fig:two districts}.

\begin{figure}[htbp]
	\centering
	\begin{subfigure}{.4\textwidth}
		\centering
		\includegraphics[width=.8\linewidth]{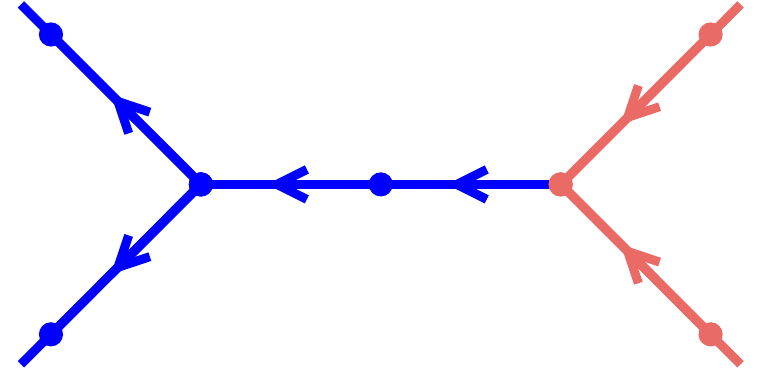}
		\label{fig:subdiv}
	\end{subfigure}
	\begin{subfigure}{.5\textwidth}
		\centering
		\includegraphics[width=.8\linewidth]{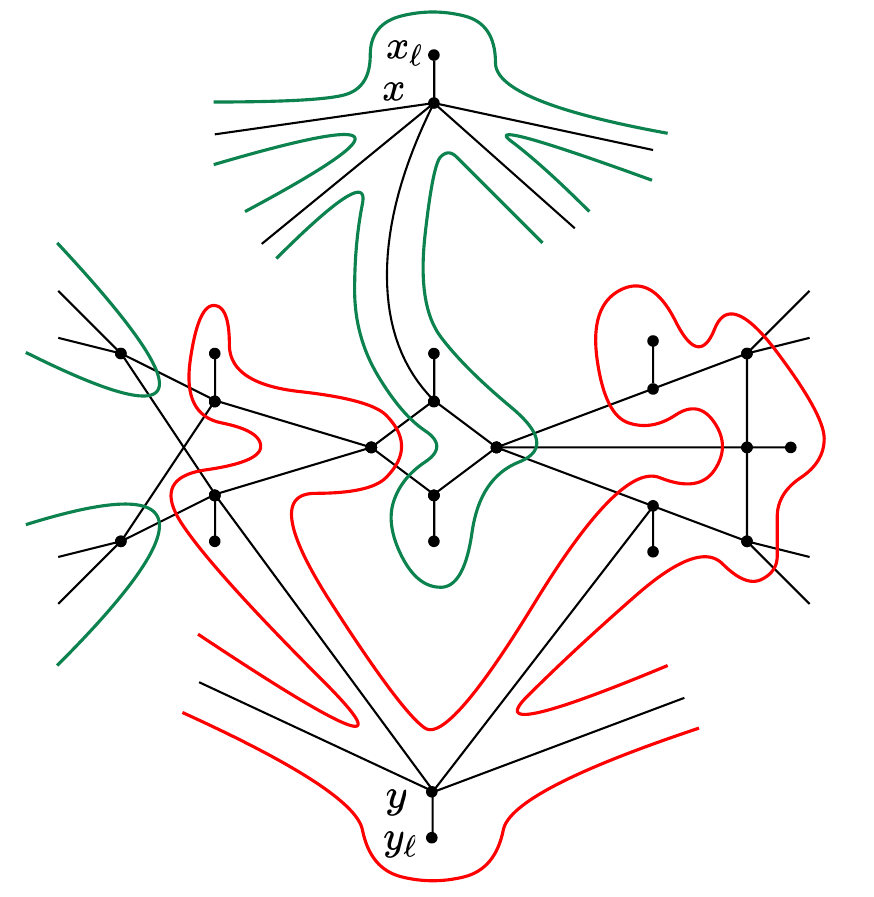}
		\label{fig:two district reduction}
	\end{subfigure}
	\caption{Transforming orientation on subdivided $G_{NCL}$ to 2-district map over $G$.
	}
	\label{fig:two districts}
\end{figure}

Now, construct $G$ from this subdivided version of $G_{NCL}$ similar to Section~\ref{ssec:PSPACE-general}  (using the original nonplanar OR gadget).
The subdivided version of $G_{NCL}$ has three types of vertices: OR, AND, and subdivision vertices. Each vertex is replaced by a gadget for the corresponding type.
Next, create a vertex $x$ with one leaf $x_\ell$ attached to it, and an edge connecting $x$ to one anchor in every degree-two vertex gadget (either anchor is fine). Then create a vertex $y$ with one leaf $y_\ell$ attached to it, and add an edge connecting $y$ to one anchor in every OR and AND gadget (again any anchor is fine).

Finally, given an orientation on $G_{NCL}$, we start by building a district map on $G$ in the same way as before, but after we have built this map, we merge all the districts on degree-two gadgets into a single district also containing $x$ and $x_\ell$, and then merge all the districts on AND and OR vertices into another single district also containing $y$ and $y_\ell$. This construction is shown in Figure~\ref{fig:two districts}.

\begin{theorem} \label{thm:two district pspace}
Given a graph $G$ and two 2-district maps $\Pi_A$ and $\Pi_B$ over $G$, it is PSPACE-complete to determine if $\Pi_A$ and $\Pi_B$ are in the same connected component of $\Gamma_2(G)$.
\end{theorem}
\begin{proof}
Since the leaves $x_\ell$ and $y_\ell$, and the leaves in every gadget must remain in their respective districts, and $x$ and $y$ are only adjacent to one anchor in each gadget, the connectivity requirements within each gadget remain, and thus every argument in Lemma~\ref{lem:connectedness satisfies} still applies.
The two partite sets of $G_{NCL}$ form the basis of the two districts in our map.
Thus, every adjacency between two gadgets in $G$ is between gadgets in different districts, and conversely adjacency between the two districts is always between two neighboring gadgets. So all of the arguments in Lemma~\ref{lem:switch to flip} still apply, yielding the stated result.
\end{proof}

\subsection{Exponential Diameter}

The following theorem is implicit in the reduction from QSAT to NCL in~\cite{PSPACE}.
\begin{theorem} \label{thm:ncl diameter}
For every $n \in \mathbb{N}$ there exist a planar NCL graph $G_{NCL}$ on $n$ vertices and initial and final orientations $A, B$ on $G_{NCL}$ such that $2^{\Theta(n)}$ edge flips are necessary and sufficient to reconfigure $A$ into $B$.
\end{theorem}
\begin{proof}
First, take the construction in Figure~4 of~\cite{PSPACE} showing an NCL graph which simulates a quantified formula evaluator. Now modify all $n$ quantifier blocks in this construction to be universally quantified. Next, build any tautological boolean formula on $n$ variables which can be expressed using a linear number of NCL gates; in particular the disjunction $x_1 \lor \overline{x_1} \lor \ldots \lor x_n \lor \overline{x_n}$ suffices. The resulting NCL graph thus has a total number of vertices and edges linear in the number of quantifiers. By Lemma 3 of the same paper we see that the $i^{th}$ universal quantifier cannot have its satisfied-out edge flipped until the remainder of the quantified formula is evaluated under both variable assignments of $x_i$, so if a reconfiguration exists it requires at least $2^{\Omega(n)}$ edge flips; and in this case a reconfiguration is possible with $2^{O(n)}$ edge flips since all quantifiers are existential and the formula is a tautology. Since the original QSAT reduction contains some edge crossings, one might worry that in deploying crossover gadgets that ensure planarity we may see a quadratic blow-up in the size of the graph, weakening our bound to $2^{\Omega(\sqrt{n})}$ in the planar case. However, by inspection we see that each universal quantifier gadget contains only three edge crossings, and the formula $x_1 \lor \overline{x_1} \lor \ldots \lor x_n \lor \overline{x_n}$ can be constructed as a perfect binary tree with no crossings, so our lower bound holds in the planar case, as well.
\end{proof}

\begin{corollary} \label{corollary:exp switch diameter}
The diameter of a connected component of $\Gamma_k(G)$ can be as large as $2^{\Omega(n)}$ where $n = |V(G)|$, even if $G$ is a planar graph of maximum degree 6, or if $k=2$.
\end{corollary}
\begin{proof}
Since our reduction creates a redistricting instance over a graph linear in the size of the original NCL graph, this result follows from the reductions in Sections~\ref{ssec:PSPACE-general}--\ref{ssec:PSPACE-two} and Theorem~\ref{thm:ncl diameter}.
\end{proof} 

\ShoLong{}{
\section{Hardness for Shortest Paths}
\label{ssec:hardnes-shortest}
In the previous section we showed that it is PSPACE-complete to decide whether two district maps on $G$ are in the same component of $\Gamma_k(G)$. Hardness here crucially relied on the fact that $\Gamma_k(G)$ can have many connected components, each with potentially exponential diameter. In this section, we show that even if we constrain $G$ to be biconnected, and thus $\Gamma_k(G)$ to be connected with polynomially bounded diameter (cf. Theorem~\ref{thm:2conn-alg}), the problem of finding a shortest path from one district map to another in $\Gamma_k(G)$ is NP-hard.
We start by showing NP-hardness for arbitrary graphs (Lemma~\ref{thm:shortest}) and then strengthen it to biconnected graphs (Theorem~\ref{thm:biconnected-hardness}).
The decision problem can be formally stated as follows: we are given a graph $G$, two $k$-district maps, $\Pi_A$ and $\Pi_B$, and an integer $L\geq 0$, and ask whether a sequence of at most $L$ switches can take $\Pi_A$ to $\Pi_B$. Let us denote this problem by $R(G,\Pi_A,\Pi_B,L)$.

We present a polynomial-time reduction from 3SAT. An instance of 3SAT consists of a boolean formula $\phi$ in 3CNF. Let $m$ and $n$ be the number of clauses and the number of variables, respectively, in $\phi$. We construct, for a given 3SAT instance $\phi$, a graph $G(\phi)$, two district maps $\Pi_A(\phi)$ and $\Pi_B(\phi)$, and a nonnegative integer $L(\phi)$. We then show that $\phi$ is satisfiable if and only if the instance $R(G(\phi),\Pi_A(\phi),\Pi_B(\phi),L(\phi))$ of the redistricting problem is positive.


\begin{figure}[!htb]
\begin{center}
\includegraphics[width=0.5\textwidth]{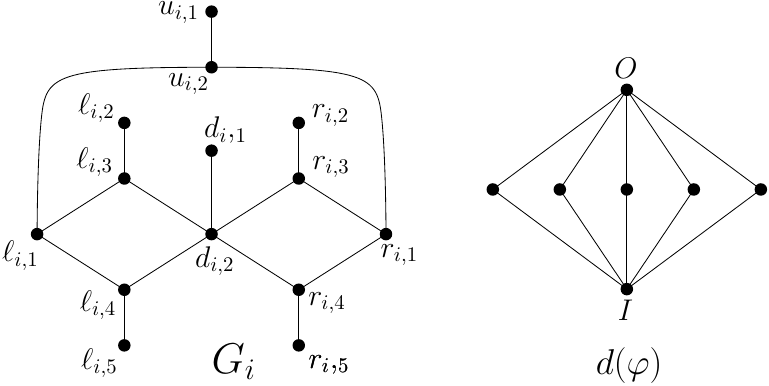}
\caption{A variable gadget (left) and a districting pipe (right)}
\label{fig:ShortestPathGadgets}
\end{center}
\end{figure}

We construct the graph $G(\phi)$ as follows:
\begin{enumerate}\itemsep 0pt
\item For every variable $x_i$, construct a \emph{variable gadget} $G_i$, shown in Figure~\ref{fig:ShortestPathGadgets}~(left).
\item For every clause $c_j$, a \emph{clause gadget} $H_j$ consists of two adjacent vertices, $c_{j,1}$ and $c_{j,2}$. See Figure~\ref{fig:graphofphi} (top).
\item For every variable $x_i$ appearing in a clause $c_j$, if $x_i$ is nonnegated (negated) in $c_j$, insert an edge between $c_{j,2}$ and $\ell_{i,1}$ ($r_{i,1}$).
\item Next, we add a subgraph, called a \textbf{districting pipe} $d(\phi)$, that consists of $m+n+1$ vertices. The districting pipe is a complete bipartite graph between a 2-element partite set $\{O,I\}$ and a $(m+n-1)$-element partite set.
    Figure~\ref{fig:ShortestPathGadgets}~(right) depicts and example where $m+n-1=5$.
\item Lastly, for each variable gadget $G_i$, insert the edges $O\ell_{i,1}$ and $Or_{i,1}$.
\end{enumerate}

\begin{figure}[htbp]
\begin{center}
\includegraphics[width=0.9\textwidth, angle=0]{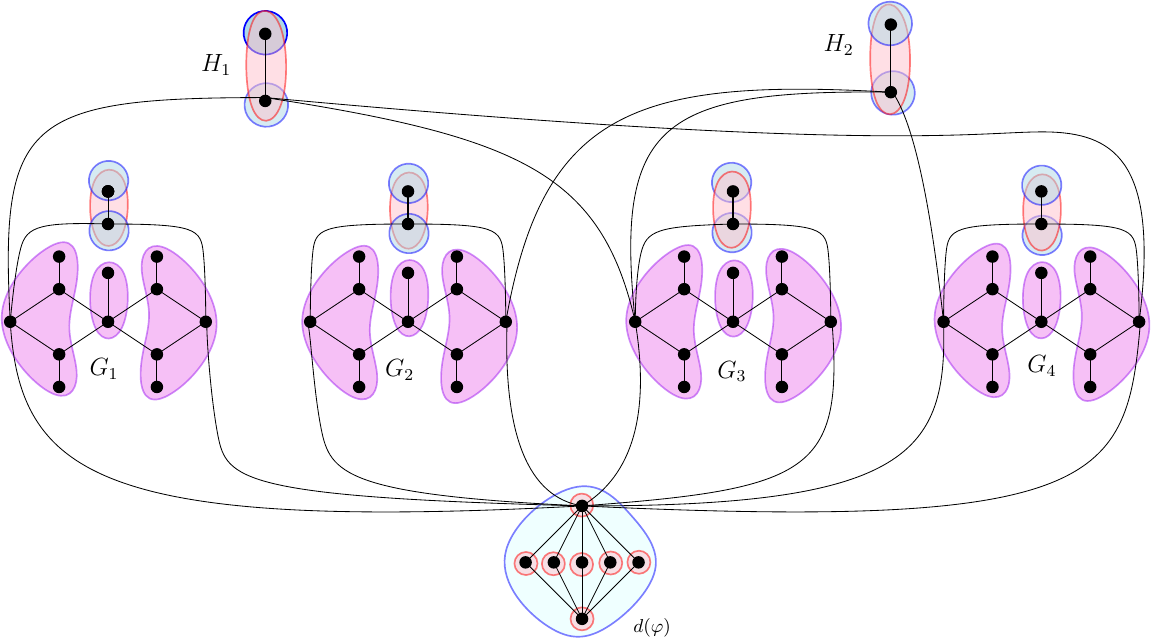}
\caption{The graph $G(\phi)$ for the formula $\phi=(x_1\lor x_3 \lor \neg x_4) \land (\neg x_2 \lor x_3 \lor x_4)$.
The red and purple regions represent the districts in $\Pi_A(\phi)$. The blue and purple regions represent the districts in $\Pi_B(\phi)$. In particular, the purple regions represent districts that are present in both $\Pi_A(\phi)$ and $\Pi_B(\phi)$.}
\label{fig:graphofphi}
\end{center}
\end{figure}

We now define two $(2m+5n+1)$-district maps on $G(\phi)$. Refer to Figure~\ref{fig:graphofphi}.
First, let $\Pi_A(\phi)$ consist of the following districts. For each variable gadget $G_i$, we create four districts:
$\ell_i=\{\ell_{i,j}: j=1,\ldots ,5\}$, $r_i=\{r_{i,j}: j=1,\ldots , 5\}$, $\{d_{i,1},d_{i,2}\}$, and $\{u_{i,1},u_{i,2}\}$.
For each clause gadget $H_j$, we create a 2-element district $\{c_{j,1},c_{j,2}\}$. In the districting pipe, every vertex is in a singleton district, which yields $m+n+1$ singletons.
Next, we define the target district map, $\Pi_B(\phi)$. For every variable gadget $G_i$, we create similar districts to $\Pi_A(\phi)$, the only difference is that the district $\{u_{i,1},u_{i,2}\}$ is now split into two singletons: $\{u_{i,1}\}$ and $\{u_{i,2}\}$.
In each clause gadget $H_j$, the two vertices form singleton districts.
Lastly, the district pipe now consists of one $(m+n+1)$-vertex district.
Finally, we set $L(\phi):=4m+7n-1$. This completes the description of the instance $R(G(\phi),\Pi_A(\phi),\Pi_B(\phi),L(\phi))$.

\begin{lemma}\label{lem:==>}
If there exists a satisfying truth assignment for $\phi$, then $R(G(\phi),\Pi_A(\phi),\Pi_B(\phi),L(\phi))$ is a positive instance.
\end{lemma}
\begin{proof}
Let $\tau$ be a satisfying truth assignment for $\phi$.
We show that $\Pi_A(\phi)$ can be transformed into $\Pi_B(\phi)$ using $L(\phi)$ switches.
We define an open gate of $G_i$ to be either $\ell_{i,1}$ if district $\ell_i$ contains the vertex $d_{i,2}$, or $r_{i, 1}$ if district $r_i$ contains $d_{i,2}$. This implies that the open gate is a vertex in the variable gadget that can be taken by external districts. 
Note that $\ell_{i}$ (resp., $r_{i}$) cannot leave vertices $\ell_{i,2}$, $\ell_{i,3}$, $\ell_{i,4}$, and $\ell_{i,5}$ (resp., $r_{i,2}$, $r_{i,3}$, $r_{i,4}$, and $r_{i,5}$) by Lemma~\ref{lem:leaves}.
So then both $\ell_{i,1}$ and $r_{i,1}$ cannot be taken by external districts, or else $\ell_{i}$ or $r_{i}$ would be disconnected.

\begin{enumerate}
\item For each variable $x_i$, if $\tau(x_i)=\texttt{true}$ ($\texttt{false}$), then expand $\ell_{i}$ ($r_i$) into $d_{i,2}$ (thereby opening one of the two gates for each variable) in a total of $n$ switches.

\item For every variable $x_i$, we transform the district containing $O$ into a district containing only $u_{i,2}$ and the open gate of $G_i$ by expanding twice (for $\tau(x_i)=\texttt{true}$, $\sw(\ell_{i,3},\ell_{i,1},O)$ and $\sw(u_{i,1},u_{i,2},\ell_{i,1})$) and shrinking it out of $O$ using either one, for $x_1$, or two switches, for remaining variables (for $\tau(x_i)=\texttt{true}$, $\sw(O,p_{i},I)$ and $\sw(\ell_{i,1},O, p_{i+1})$ where $p_i$ and $p_{i+1}$ are vertices in the districting pipe). 
We can perform this step in $3n+(n-1)$ switches by, for the previously mentioned shrinks, expanding a singleton containing a degree-2 vertex of $d(\phi)$ whenever possible, or expanding the district containing $I$.

\item For each clause $c_j$, choose any variable $x_i$ that appears in a true literal in $c_j$ (guaranteed to exist by the definition of $\tau$). Transform the district containing $O$ into a district containing only $c_{j,2}$ and the open gate of $x_i$, similar to the last step. Using the same strategy, this step can be accomplished with $4m$ switches.
After this step $m+n$ districts were moved out of $d(\phi)$, and now a single district contains all its vertices (the district initially containing $I$).

\item Finally, for every variable $x_i$ we close its gate by first expanding either $\ell_i$ or $r_i$ into the open gate and then expanding the singleton district at $d_{i,1}$ into $d_{i,2}$. This takes $2n$ switches.
\end{enumerate}

Overall, we have performed $n+3n+(n-1)+4m+2n=4m+7n-1=L(\phi)$ switches. These $L(\phi)$ switches transformed $\Pi_A(\phi)$ to $\Pi_B(\phi)$, and so the instance $R(G(\phi),\Pi_A(\phi),\Pi_B(\phi),L(\phi))$ is positive.
\end{proof}

\begin{lemma}\label{lem:<==}
If $R(G(\phi),\Pi_A(\phi),\Pi_B(\phi),L(\phi))$ is a positive instance of the redistricting problem for a boolean formula $\phi$, then there exists a satisfying truth assignment for $\phi$.
\end{lemma}
\begin{proof}
We derive lower bounds on the number of switches in any sequence of switches from $\Pi_A(\phi)$ to $\Pi_B(\phi)$ by making inferences from the initial and target district maps. Notice that if a district contains a leaf, then the leaf remains in the same district by Lemma~\ref{lem:leaves}. We call a district \emph{mobile} if it does not contain any leaf in $\Pi_A(\phi)$. By construction, only the $m+n+1$ districts initially in the districting pipe are mobile.

\begin{itemize}
\item[(A)] Since $u_{i,1}$ and $u_{i,2}$ are in distinct districts in $\Pi_B(\phi)$, we must have a mobile district that travels to $u_{i,2}$. In order to accomplish this, we must first open one of the two gates of the variable gadget $G_i$. Opening $n$ gates, one in each variable gadget, requires at least $n$ switches.

\item[(B)] As noted above, a mobile district must travel to $u_{i,2}$ for $i=1,\ldots , n$. Moving $n$ mobile districts from $O$ to $u_{i,2}$, $i=1,\ldots , n$, requires at least $2n$ switches, and an additional $2(n-1)$ switches for $n-1$ mobile districts to reach $O$: one switch expanding the district containing $I$ and one expanding the desired mobile district to $O$. Overall, this requires at least $2n+2(n-1)$ switches.

\item[(C)] Since each clause gadget $H_j$ consists of two districts in $\Pi_B(\phi)$, a mobile district from the districting pipe must travel to $c_{j,2}$, for $j=1,\ldots, m$, which requires $4m$ switches.

\item[(D)] Because one mobile district must expand to the entire district pipe $d(\phi)$, either one mobile district expands into $I$, or the district that contains $I$ expands into $O$. In either case, this takes one additional move that has not been counted so far.

\item[(E)] Note that the gate of $G_i$ is closed and $\{d_{i,1},d_{i,2}\}$ is a 2-vertex district in $\Pi_B(\phi)$, for $i=1,\ldots n$. So the district of the open gate must expand to consume its gate, and the singleton district at the leaf $d_{i,1}$ must expand into $d_{i,2}$. Together this requires a total of $2n$ switches.
\end{itemize}

Therefore, we need at least $n+2n+2(n-1)+4m+1+2n=7n+4m-1=L(\phi)$ switches to solve the redistricting problem. Since we executed exactly $L(\phi)$ switches, no other move is allowed.

Due to the fact that after opening a gate, opening the opposite gate would require additional switches, we conclude that precisely one gate opens in each variable gadget. We construct a truth assignment as follows: For every $i=1,\ldots, n$, let $\tau(x_i)=\texttt{true}$ if the left gate of the variable gadget $G_i$ opens, and $\tau(x_i)=\texttt{false}$ otherwise. Since the only way to get a district to $c_{j,2}$ was through
an open gate of one of the three literal in the clause $c_j$, then every clause is incident to an open gate of a variable gadget. Since every open gate corresponds to a true literal, at least one of the three literals is true in each clause. Henceforth, $\tau$ is a satisfying truth assignment for $\phi$.
\end{proof}

The following Lemma is the direct consequence of Lemmas~\ref{lem:==>} and \ref{lem:<==}.

\begin{lemma}\label{thm:shortest}
Given an graph $G$ with $n$ vertices and an integer $1 \leq k\leq n$, it is NP-complete to decide whether the length of a shortest path in $\Gamma_k(G)$ between two given district maps is below a given threshold.
\end{lemma}

The unshrinkable districts in the previous reduction are not essential for NP-hardness.
We can modify the reduction to produce a biconnected graph $G$ as follows. 

\begin{theorem}
\label{thm:biconnected-hardness}
It is NP-hard to decide whether the length of a shortest path in $\Gamma_k(G)$ between two given shrinkable district maps is below a given threshold even when $G$ is biconnected.
\end{theorem}
\begin{proof}
In the reduction above, for a boolean formula $\phi$ in 3CNF, we constructed an instance $R(G(\phi),\Pi_A(\phi),\Pi_B(\phi),L(\phi))$ of the redistricting problem. 
We modify the reduction by subdividing the edge connected to every leaf creating a path of length $L(\phi)$.
Now connect every leaf to the vertex $I$ in $d(\phi)$.
The resulting graph is biconnected because the only cut vertices produced in the previous reductions were either $O$ in $d(\phi)$ or adjacent to leaves. 
The modification guarantees that they are no longer cut vertices.
We make the following modifications to the district maps $\Pi_A(\phi)$ and $\Pi_B(\phi)$: If both endpoints of a subdivided edge belonged to the same district, add the new vertex created by the subdivision to this district; Extend the singleton districts that previously contained a leaf in $B(\phi)$ to contain all new vertices on its path. 
The districts that contain long paths (of length $L(\phi)$) cannot leave such paths completely since we are only allowed $L(\phi)$ switches.
Then, the only way to obtain $\Pi_B(\phi)$ from $\Pi_A(\phi)$ is to move the singleton districts in $d(\phi)$ through the variable gadgets.
The rest of the proof is analogous to the proof of Lemmas~\ref{lem:==>} and \ref{lem:<==}.
\end{proof}

}

\section{Conclusion}
\label{sec:con}

This paper provides the theoretical foundation for using elementary switch operations to explore the configuration space $\Gamma_k(G)$ of all partitions of a given graph into $k$ nonempty subgraphs, each of which is connected. We gave a polynomial-time testable combinatorial characterization for connected configurations spaces (Theorem~\ref{thm:conn-test}). 

Our PSPACE-hardness proof with few (two) districts produces a nonplanar graph. The complexity of deciding whether two $k$-district maps (with $k=O(1)$) on a planar graph $G$ are in the same component of $\Gamma_k(G)$ remains open. 
Our NP-hardness reduction for the shortest path problem produces a nonplanar (biconnected) graph.
It is an open problem to determine the computational complexity of computing shortest paths in $\Gamma_k(G)$ when $G$ is biconnected and planar, or in $\Gamma_k'(G)$ when $G$ is planar.

A crucial concept in both the combinatorial characterization and the reconfiguration algorithms (Algorithms~\ref{algo:2conn} and~\ref{algo:1conn}) was \emph{\contractibility}: A district is \contractible\ if it can be reduced to a single vertex (while all $k$ districts remain connected).
In applications to electoral maps, all districts have roughly average size, say between $\frac{n}{2k}$ and $\frac{2n}{k}$, and a singleton district is impractical. 
In a sense, we establish that there is a path between any two \contractible\ district maps with average-size districts by passing through ``impractical'' district maps with singleton districts. We do not know whether singleton districts are necessary: for a constant $c\geq 1$, we can define $\Gamma_{k,c}(G)$ as the graph of $k$-district maps in which the size of every district lies in the interval $[\frac{n}{ck},\frac{cn}{k}]$. It is easy to construct examples where $\Gamma_{k,1}(G)$ has isolated vertices. Is there a constant $c>1$ such that the connectedness of $\Gamma_k(G)$ implies that $\Gamma_{k,c}(G)$ is also connected?


In our model, a district map is a partition of the vertex set into $k$ \emph{unlabeled} nonempty subsets. One could consider the \emph{labeled} variant, and define a switch graph $\Gamma^L_k(G)$ on labeled $k$-district maps. Our results do not carry over to this variant: in particular, the labeled switch graph $\Gamma_k^L(G)$ need not be connected if $G$ is biconnected. For example, if $G=C_n$ (i.e., a cycle of $n\geq 3$ vertices) and $k\geq 3$, then the cyclic order of the districts along the cycle cannot change. In the special case that $k=2$ and $G$ is biconnected, $\Gamma_2^L(G)$ is connected since we can \contract\ a district to a singleton (cf.~Lemma~\ref{lem:con2}) and move it to any vertex while the complement remains connected. 
When we move a singleton district from one vertex to another, it temporarily occupies both vertices, which should not form a 2-cut.
\Contract ing a district to a singleton is sometimes necessary in this case (one such example is $G=K_{2,m}$, $m\geq 3$, where the 2-element partite set is split between the two districts).

\paragraph*{Acknowledgments.} 
We are grateful to Jesse Geneson for pointing out an error in an earlier version of Lemma~\ref{lem:conn-test}. 

\bibliographystyle{plainurl}
\bibliography{redistricting}

\ShoLong{\newpage
\appendix

\section{Block Trees and SPQR Trees}
\preliminarysec

\section{Omitted proofs}

\subsection{Proof of Lemma \ref{lem:con2}}
\label{sec:proofLemConTwo}
\proofLemConTwo

\subsection{Proofs from Section \ref{sec:alg}}

\pfclaimAAA
\pfclaimBBB
\pfclaimCCC

\subsection{Proof of Lemma \ref{lem:conn-char}}
\pflemconnchar

\section{Lower Bounds for Shortest Paths}
\label{sec:LB}

In this section, we prove lower bounds for the diameter of the switch graph $\Gamma_k(G)$, and for the length of a shortest path in $\Gamma_k(G)$ (when $\Gamma_k(G)$ need not be connected). We start with a simple construction that yields an $\Omega(kn)$ lower bound for the diameter of $\Gamma_k'(G)$ when $k\leq n/2$, which matches the upper bound of Theorem~\ref{thm:general-graphs}.

\begin{theorem}\label{thm:contractible-LB}
For every $k,n\in \mathbb{N}$, $k\leq n$, there exists a graph $G$ with $n$ vertices
such that the diameter of  $\Gamma_k'(G)$ is $\Omega(k(n-k))$.
\end{theorem}
\begin{proof}
Let $G$ be a path $(v_1,\ldots, v_n)$ with $n$ vertices. Let $\Pi_1$ consist of $V_i=\{v_i\}$ for $i=1,\ldots , k-1$, and $V_k=\{v_k,\ldots, v_n\}$; and let $\Pi_2$ be the partition $W_1=\{v_1,\ldots v_{n-k}\}$ and $W_j=\{v_{n-k+j}\}$, for $j=2,\ldots, k$.
Assume that a sequence of switch operations takes $\Pi_1$ to $\Pi_2$. Since each district is nonempty at all times,
district $V_i$ is transformed into $W_i$ for all $i\in \{1,\ldots, n\}$. Each switch operation can move the rightmost vertex of at most one district, and by at most one unit. The rightmost vertex of $V_k$ remains fixed. The distance traveled by the rightmost vertices of $V_1,\ldots, V_k$ each is $n-k$. This requires at least $(k-1)(n-k)$ operations.
\end{proof}

If we connect the two endpoints of the path we get a cycle (and in particular a biconnected graph). The lower bound of Theorem~\ref{thm:contractible-LB} can be adapted to this case, which in particular implies that the diameter is not reduced even when $G$ is biconnected.

\begin{theorem}\label{thm:contractible-LB+}
For every $k,n\in \mathbb{N}$, $k\leq n$, there exists a biconnected graph $G$ with $n$ vertices
such that the diameter of  $\Gamma_k(G)'$ is $\Omega(k(n-k))$.
\end{theorem}
\begin{proof}
Let $G=C_n$ be the cycle with $n$ vertices $(v_1,\ldots, v_n)$.
We construct two $k$-district maps, $\Pi_1$ and $\Pi_2$. Let $\Pi_1$ consist of $V_i=\{v_i\}$ for $i=1,\ldots , k-1$, and $V_k=\{v_k,\ldots, v_n\}$. The partition $\Pi_2$ is the copy of $\Pi_1$ rotated by $\lfloor n/2\rfloor$, that is, $W_i=\{v_{i+\lfloor n/2\rfloor}\}$ for $i=1,\ldots , k-1$, and $W_k=\{v_{k+\lfloor n/2\rfloor},\ldots, v_{n+\lfloor n/2\rfloor}\}$, where we use arithmetic modulo $n$ on the indices.

Assume that a sequence of switch operations takes $\Pi_1$ to $\Pi_2$. Note that the cyclic order of the district cannot change,
and so there is an integer $r\in \{0,\ldots , k-1\}$ such that $V_i$ is transformed to $W_{i+r\mod k}$ for all $i\in \{1,\ldots ,k\}$.
For any $r$, at least $k-2$ districts are singletons in both $\Pi_1$ and $\Pi_2$. The sum of the shortest distances
along $C_n$ between the initial and target positions is a lower bound for the number of switches.

If $r\leq \lfloor k/2\rfloor$, then the shortest distance between the initial and target positions is at least $\lfloor n/2\rfloor -r\leq \Omega(n-k)$ for the districts $V_i$, $i=1\ldots, k-1-r$; which sums to $\Omega(k(n-k))$.
If $\lfloor k/2\rfloor <r<k$, then shortest distance is at least $\lfloor n/2\rfloor -(k-r)\leq \Omega(n-k)$ for $V_i$, $i=r,\ldots, k-1$; which also sums to $\Omega(k(n-k))$.
\end{proof}

In the remainder of this section, we establish lower bounds for the diameter of a single component of $\Gamma_k(G)$, when $\Gamma_k(G)$ is disconnected (cf.~Theorem~\ref{thm:conn-test}).

\subsection{Diamonds}
\label{ssec:diamond}

A diamond is a useful construction for hardness reductions for the redistricting problem (cf.~Section~\ref{sec:hardness}). The simplest case of a diamond, with six vertices, is depicted in Figure~\ref{fig:diamond}. A diamond has two leaves: If both leaves are in the same district, then this district is incontractible by Lemma~\ref{lem:con1}, and we can encode the truth value of a variable by one of two possible paths between the leaves.

\begin{figure}[!htb]
\begin{center}
\includegraphics[width=0.35\textwidth]{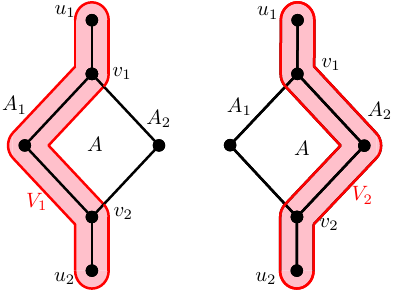}
\caption{A diamond. If a district $V_i$ contains both leaves, $u_1$ and $u_2$, then $V_i$ contains a path between the leaves.}
\label{fig:diamond}
\end{center}
\end{figure}

Formally, a \emph{diamond} $A$ is an induced subgraph of $G$ with two cut vertices $v_1,v_2\in C(G)$ that are connected by with $m\geq 2$ interior-disjoint paths $A_1,\ldots ,A_m$, where $m-1$ of these paths contain exactly one interior node; and $v_1$ (resp., $v_2$) is the  endpoint of a dangling path whose other endpoint is a leaf $u_1$ (resp., $u_2$).
Two crucial properties of diamonds are formulated in the following observation.

\begin{observation}\label{obs:cut-containment}
Let $V_i$ be a district that contains both leaves $u_1$ and $u_2$.
\begin{enumerate}\itemsep -2pt
\item\label{obs1:cut-containment} Then $V_i$ contains both cut vertices, $v_1,v_2\in C(G)$, and all vertices of some path $A_1,\ldots , A_m$.
\item\label{obs2:path-availability} In a sequence of switches, if $V_i$ contains path $A_j$ and later $A_{j'}$, $j\neq j'$, then
    $V_i$ contains two disjoint paths between $v_1$ and $v_2$ in some intermediate state.
\end{enumerate}
\end{observation}

We use Observation~\ref{obs:cut-containment}\eqref{obs2:path-availability} repeatedly in the hardness reductions, typically, with $m=2$ or $m=3$.


For $r\in \mathbb{N}$, a \emph{diamond chain} of length $r$ is a subgraph of $G$ with $r$ diamonds, $A_1,\ldots A_r$, where each diamond has precisely two interior-disjoint paths between its cut vertices (a \emph{left} path and a \emph{right} path), with one interior node each, and the interior vertex in the right path of $A_i$ is identified with the one in the left path of $A_{i+1}$ for $i=1,\ldots , r-1$.
Denote the interior vertex in the left path of $A_i$ by $a_i$ for $i=1,\ldots , r$; and the interior vertex of the right path of $A_r$ by $a_{r+1}$. A chain of length 6 is depicted in Figure~\ref{fig:chain}.

\begin{figure}[!htb]
\begin{center}
\includegraphics[width=0.5\textwidth]{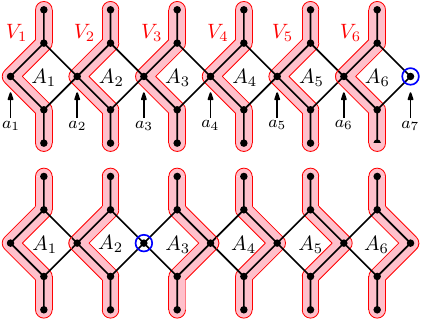}
\caption{A chain of 6 diamonds with two different district maps,
each containing an external singleton at a different point.}
\label{fig:chain}
\end{center}
\end{figure}

Assume that a district map $\Pi$ in which district $V_i$ contains the two leaves of diamond $A_i$, for $i=1,\ldots r$.
\begin{observation}\label{obs:unique-singleton}
In a chain of diamonds, there at most one vertex that is not in any of the districts $V_1,\ldots , V_r$,
and it is one of the vertices $a_1,\ldots a_{r+1}$.
\end{observation}
\begin{proof}
By Observation~\ref{obs:cut-containment}, each district $V_i$, $i\in \{1,\ldots ,r\}$ contains at least five vertices: both leaves and both cut vertices in the diamond $A_i$, and an interior vertex in at least one path between the cut vertices. Since the chain of diamonds has only $5r+1$ vertices, at most one vertex belongs to some other district, and such a vertex is neither a leaf nor a cut vertex.
\end{proof}

We say that the diamond $A_i$ is \emph{switched to the left} (resp., \emph{right}) if $a_i\in V_i$ (resp., $a_{i+1}\in V_i$).
By Observation~\ref{obs:cut-containment}, every diamond in the chain is switched to left or right (or both).
We show that the switch position in one diamond determines all others to its left or right.

\begin{lemma}\label{lem:cascade-switch}
If a diamond $A_i$ is switched to the left ($a_i\in V_i$), then $A_j$ is also switched to the left ($a_j\in V_j$) for all $j\{1,\ldots, i-1\}$. Likewise, if $A_i$ is switched to the right ($a_{i+1}\in V_i$), then $A_j$ is also switched to the right ($a_{j+1}\in V_j$) for all $j\in\{i+1,\ldots, K\}$.
\end{lemma}
\begin{proof}
If $A_i$ is switched to the left ($a_i\in V_i$), then $A_{i-1}$ cannot be switched to the right because the districts are disjoint.
By  Observation~\ref{obs:cut-containment}\eqref{obs2:path-availability}, $A_{i-1}$ must be switched to the left $(a_{i-1}\in V_{i-1}$).
Induction completes the proof. An analogous argument applies when $A_i$ is switched to the right.
\end{proof}

\begin{lemma}\label{lem:switches-across-chain}
Let $\Pi_1$ and $\Pi_2$ be two district maps such that
$a_i\notin \bigcup_{x=1}^r V_x$ in $\Pi_1$ and
$a_j\notin \bigcup_{x=1}^r V_x$ in $\Pi_2$.
Then the distance between $\Pi_1$ and $\Pi_2$ is at least $|j-i|$ in $\Gamma_k(G)$.
\end{lemma}
\begin{proof}
Assume, without loss of generality, that $i<j$. Then the diamonds between $A_i,\ldots, A_{j-1}$ are switched to the right in $\Pi_1$ and to the left in $\Pi_j$ by Lemma~\ref{lem:cascade-switch}. These $j-1$ diamonds must switch. Each diamond requires one switch to expand into its left path (the contraction from the right path not counted, since each switch contracts one district and expands another).
\end{proof}

\subsection{Incontractible}
\label{ssec:LB-incontractible}

In this section, we construct a graph $G$ and two $k$-district maps $\Pi_1$ and $\Pi_2$,
and show that they are at distance distance $\Omega (k^3 + kv)$ in $\Gamma_k(G)$.

\begin{theorem}\label{thm:LB}
For every $k,n\in \mathbb{N}$, $6k\leq n$, there exists a graph $G$ with $n$ vertices and two $k$-district maps, $\Pi_1$ and $\Pi_2$, such that the switch graph $\Gamma_k(G)$ contains a path between $\Pi_1$ and $\Pi_2$, but the length of every such path is $\Omega(k^3+kn)$.
\end{theorem}
\begin{proof}
We construct a graph $G$ in terms of three parameters, $r$, $q$, and $\ell$, and choose their values at the end of the proof.
Create two disjoint chains of diamonds, $A=(A_1,\ldots, A_r)$ and $B=(B_1,\ldots ,B_{r-1})$, of lengths $r$ and $r-1$, respectively, where $r$ is an even integer to be specified later. Denote the interior vertices of the paths in the diamonds by
$a_1,\ldots , a_{r+1}$ and $b_1,\ldots, b_r$, respectively. Insert a (spiral) path $S$ of length $2r$ on these vertices constructed as follows: Connect $a_1$ to $b_{r/2}$. For $i\in \{2,\ldots , r/2 + 1\}$, connect $a_i$ to $b_{r/2+i-1}$ and $b_{r/2-i+1}$. For $i\in \{r/2 + 2,\ldots,  r + 1\}$, connect $a_i$ to $b_{i-1-r/2}$ and $b_{3r/2-i-2}$. Connect $a_{r/2+1}$ to $b_r$.
Finally, create two trees, $D_1$ and $D_2$, each consists of a path of length $q$ and
$\ell$ leaves attached to one endpoint; and identify the other endpoints of the paths
with $a_1$ and $a_{r/2+1}$, respectively.
See Fig.~\ref{fig:fans} for an example.
The total number of vertices is $n=5(2r-1)+2+2\ell+2q-2=10r+2\ell+2q-5$.

\begin{figure}[!htb]
\begin{center}
\includegraphics[width=0.55\textwidth]{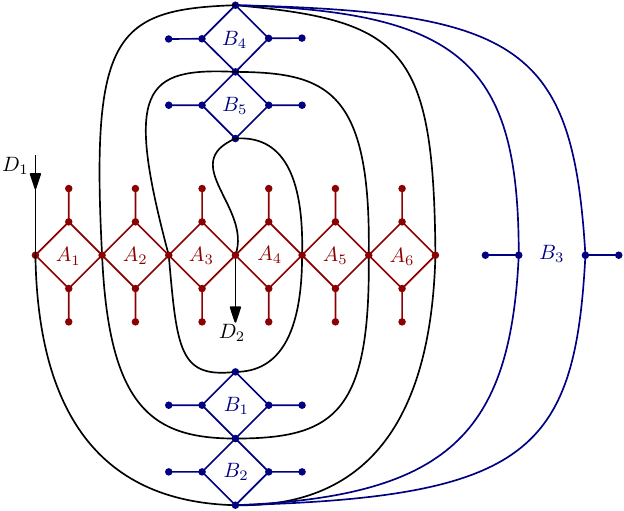}
\caption{Graph $G$ for $r=6$, with two chains of diamonds, and a spiral $S$.}
\label{fig:incont-lower}
\end{center}
\end{figure}

\begin{figure}[!htb]
\begin{center}
\includegraphics[width=0.65\textwidth]{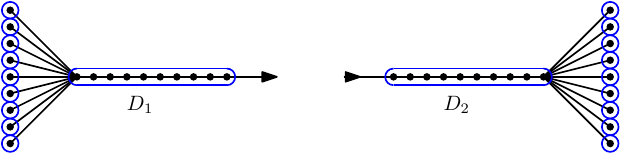}
\caption{$D_1$ in $\Pi_1$, and $D_2$ in $\Pi_2$}
\label{fig:fans}
\end{center}
\end{figure}

We construct two district maps, $\Pi_1$ and $\Pi_2$ on $G$. Both $\Pi_1$ and $\Pi_2$ have $2r - 1$ diamond districts switched to the left (i.e., $a_i\in V_i$ and $b_i\in W_i$), and the last diamond ($A_r$ and $B_{r-1}$) are switched to both left and right.
These districts contain all vertices of the two chains of diamonds.
The district map $\Pi_1$ contains $\ell$ singleton districts at the leaves of $D_1$
and one district for the path in $D_1$, and one district for all vertices of $D_2$;
while $\Pi_1$ contains one district in $D_1$, $\ell$ singleton districts
at the leaves of $D_2$, and one in the path in $D_2$.
The total number of districts is $k=(2r-1)+\ell+2$.

Consider a shortest sequence $\sigma$ of switches that takes $\Pi_1$ to $\Pi_2$
(we later show that such a sequence exists). The sequence $\sigma$ defines
a one-to-one correspondence between the districts in $\Pi_1$ and $\Pi_2$.
Both $\Pi_1$ and $\Pi_2$ contain $2r-1$ identical districts in diamonds,
that each contain two leaves. By Lemma~\ref{lem:con1}, these district
are fixed in the one-to-one correspondence.
We call the remaining $\ell+2$ districts \emph{mobile}, and conclude that
$\sigma$ moves the mobile districts in $\Pi_1$ to mobile districts in $\Pi_2$.
In particular, at least $\ell$ mobile districts move from $D_1$ to $D_2$.

\begin{claim}\label{cl:shortest-path}
If a mobile district moves from $D_1$ to $D_2$, then it moves along the spiral $S$.
\end{claim}
\begin{proof}
By Observation~\ref{obs:unique-singleton}, at most one mobile district may reside in a diamond chain at a time, but cannot travel along the edges in that chain because all adjacent vertices are cut vertices. Thus the only available edges for switches are the edges of the spiral $S$.
\end{proof}
\begin{claim}\label{cl:singleton-at-a-time}
At most one mobile district is outside of $D_1$ and $D_2$ at any time.
\end{claim}
\begin{proof}
Suppose, to the contrary, that two mobile districts are outside of $D_1$ and $D_2$ at the same time.
Then they must be in different chains, by Observation~\ref{obs:unique-singleton}. Neither can move until one of them leaves to $D_1$ or $D_2$, because they prevent the diamonds in the chain from switching between left and right (cf.~Observation~\ref{obs:cut-containment}).
However, then one of the two mobile districts would have to move back to its previous position. Therefore, we could eliminate two switches from $\sigma$, contradicting the assumption that $\sigma$ is a shortest sequence that takes $\Pi_1$ to $\Pi_2$.
\end{proof}
This claim conveniently means that the shortest sequence of switches from $\Pi_1$ to $\Pi_2$
moves each mobile district through the spiral $S$ independently.

Note that there exists a sequence of switches that takes $\Pi_1$ to $\Pi_2$.
Indeed, it is enough to move $\ell$ mobile districts from $D_1$ to $D_2$.
We can move them, one at a time, along the spiral $S$.
Since the spiral $S$ alternates between chain $A$ and chain $B$,
when the mobile district is a singleton in one chain, we can reconfigure
the other chain to make the next vertex of $S$ available.

\smallskip\noindent\textbf{Analysis.}
We analyze the length of the sequence of switches $\sigma$ that takes $\Pi_1$ to $\Pi_2$.
By Claims~\ref{cl:shortest-path} and~\ref{cl:singleton-at-a-time}, we may assume that one mobile district moves from $D_1$ to $D_2$ along $S$, and then multiply by $\ell$, the number of singletons. We further break down the cost by analyzing the number of switches needed to move a mobile district from $a_i$ to $b_{r/2+i}$, for $i < r/2$.

As a mobile district travels through $S$, it visits all vertices in $\{a_1,\ldots, a_{r+1}\}$.
Between a visit to $a_i$ and $a_j$, at least $|j-i|$ diamonds in the chain $A$ must be reconfigured, using at least $|j-i|$ switches.
The summation of gaps between consecutive vertices in in $\{a_1,\ldots, a_{r+1}\}$ is an arithmetic progression that decreases from
$r$ to 1, and sums to $\sum_{i=1}^r i=\binom{r+1}{2}$. Similarly, the summation of gaps between consecutive vertices in in $\{b_1,\ldots, b_{r}\}$ is an arithmetic progression that increases from $1$ to $r-1$, and sums to $\sum_{i=1}^{r} i=\binom{r}{2}$.
The switches that reconfigure diamonds in the chain $A$ do not have any impact in the chain $B$.

Furthermore, at least $\ell-1$ mobile districts move from some leaves of $D_1$ to some leaves of $D_2$.
These mobile district have to travel through the two paths of length $q$ in $D_1$ and $D_2$, respectively.
Any switch that expands a mobile district inside $D_1$ or $D_2$ is distinct from switches that reconfigure diamonds.
Consequently, the number of switches in $\sigma$ is at least
\[
|\sigma|\geq \ell\left(\binom{r+1}{2}+\binom{r}{2}\right)+(\ell-1)2q = \Theta(\ell r^2+\ell q).
\]
Let us choose the parameters so that $q=\Theta(n)$ and $r,\ell=\Theta(k)$, and then
$|\sigma|\geq \Omega(k^3+kn)$, as claimed.
\end{proof}

}{}

\end{document}